\nonstopmode

\documentclass[sigplan,10pt]{acmart}
\renewcommand\footnotetextcopyrightpermission[1]{}

\acmConference[]{}{}{} 
\acmBooktitle{}         
\acmYear{}              
\acmPrice{}             
\acmISBN{}              
\acmDOI{}               

\usepackage{tikz}
\usepackage{amsmath}
\usepackage{graphicx}  
\usepackage{algorithm}
\usepackage{algorithmic}
\usepackage{amsfonts}
\usepackage{booktabs}
\usepackage{stfloats}
\usepackage{multirow}
\usepackage{tikz}
\usepackage{pgfplots}
\pgfplotsset{compat=1.18}
\usepgfplotslibrary{groupplots}
\usepackage{filecontents}
\usepackage{tabularray}
\usepackage{subfigure}
\usepackage{subcaption}

\begin{document}

\date{}

\title{\Large \bf Robust Synchronisation for Federated Learning in The Face of Correlated Device Failure}

\author{
{\rm Stefan Behfar$^{1}$, Richard Mortier$^{1}$}\\
{\small $^{1}$Computer Lab, University of Cambridge, Cambridge, United Kingdom}
}

\settopmatter{printacmref=false}

\begin{abstract}
Probabilistic Synchronous Parallel (PSP) is a technique in distributed learning systems to reduce synchronization bottlenecks by sampling a subset of participating nodes per round. In Federated Learning (FL), where edge devices are often unreliable due to factors including mobility, power constraints, and user activity, PSP helps improve system throughput. However, PSP has a key limitation: it assumes device behavior is static and different devices are independent. This can lead to unfair distributed synchronization, due to highly available nodes dominating training while those that are often unavailable rarely participate and so their data may be missed.
If both data distribution and node availability are simultaneously correlated with the device, then both PSP and standard FL algorithms will suffer from persistent under-representation of certain classes or groups resulting in inefficient or ineffective learning of certain features.

We introduce Availability-Weighted PSP (AW-PSP), an extension to PSP that addresses the issue of co-correlation of unfair sampling and data availability by dynamically adjusting node sampling probabilities using real-time availability predictions, historical behavior, and failure correlation metrics. A Markov-based availability predictor distinguishes transient \emph{vs} chronic failures, while a Distributed Hash Table (DHT) layer decentralizes metadata, including latency, freshness, and utility scores.
We implement AW-PSP and trace-driven evaluation shows that it improves robustness to both independent and correlated failures, increases label coverage, and reduces fairness variance compared to standard PSP. AW-PSP thus provides an availability-aware, and fairness-conscious node sampling protocol for FL deployments that will scale to large numbers of nodes even in heterogeneous and failure-prone environments.
\end{abstract}

\maketitle

\fancyhead{}            
\fancyfoot[C]{\thepage} 

\section{Introduction}

Traditional approaches to synchronization in federated learning include Bulk Synchronous Parallel (BSP) and Asynchronous Parallel (ASP). BSP requires all workers to provide updates in one round before proceeding to the next, ensuring deterministic correctness but   risking significant delays due to slow or unreliable nodes (termed \emph{stragglers})~\cite{valiant1990bridging}. In contrast, ASP requires no synchronization at all between nodes, allowing updates to occur as quickly as possible but risking stragglers introducing stale gradients making convergence very unstable~\cite{koloskova2022sharper}. A middle ground is Stale Synchronous Parallel (SSP), which allows bounded delays before synchronization, reducing the impact of stragglers in BSP as well as the problem of stale gradients in ASP. However, it still requires global state tracking, which becomes increasingly difficult at scale~\cite{ho2013more}.

Many have explored ways to mitigate the impact of stragglers in Federated Learning (FL). Gradient coding~\cite{tandon2017gradient} reduces the dependency on slow nodes by redundant encoding of gradients so the system can recover the full gradient even when some workers lag or fail. However, this targets synchronous training setups and does not address the unpredictable availability of nodes in federated environments. Sageflow~\cite{sageflow2021} dynamically adjusts weight updates from clients when aggregating updates, mitigating the negative impact of unreliable or malicious nodes. While effective, its reliance on robust aggregation can fail to mitigate delays due to intermittent node participation. FLuID~\cite{fluid2023} uses invariant dropout to selectively ignore contributions from slow clients while maintaining model convergence. However,  ignoring slow nodes altogether can reduce model diversity and hinder convergence in heterogeneous environments.

Some FL frameworks such as Flower~\cite{Flower} and Flame~\cite{Flame} ignore node availability. Others, e.g.,~FedScale~\cite{FedScale}, Oort~\cite{Oort}, and Google's token-based availability model~\cite{GoogleFL}, have considered it but either rely on historical availability traces or perform system check-ins while empirical studies report that node availability fluctuates significantly, between 10\%---80\% across different datasets~\cite{FedScale,Flint,GoogleFL}, and FL systems that ignore this can experience a drop in accuracy of up to 50\% in real-world scenarios~\cite{workshopFL}. Probabilistic Synchronous Parallel (PSP) took a different approach, sampling a random subset of nodes to participate in the synchronization barrier for a given round. This eliminates the need for a central state controller, reduces bandwidth overheads thereby increases scalability as more nodes are involved, and significantly reduces the impact of stragglers as they can simply be ignored without weakening probabilistic convergence guarantees~\cite{wang2017probabilistic}.

However, all the above approaches suffer from (\emph{i})~participation unfairness, and (\emph{ii})~data availability co-correlation.
Participation unfairness arises when always-available nodes (e.g.,~desktop clients or stable cloud servers) dominate the training process, while intermittently available or mobile devices are rarely selected. This skewed participation reduces the diversity of training data, degrades personalization performance, and can amplify societal or geographic biases in the resulting models. Traditional  approaches such as uniform random sampling (e.g.,~PSP~\cite{wang2017probabilistic}), check-in based (e.g.,~Google FL~\cite{GoogleFL}), or trace-replay strategies (e.g.,~FedScale~\cite{FedScale}, Oort~\cite{Oort}) fail to address this imbalance effectively.

Data availability co-correlation is a challenge because devices may become unavailable in correlated patterns for many reasons, from using shared infrastructure (e.g.,~the same WiFi network or power supply) to being subject to shared user behaviors (e.g.,~devices being turned off at the same time due to sleep patterns). Such correlations in device unavailability lead to correlations in missing data as those devices are likely to have data with similar label distributions (e.g.,~when trying to learn features from local processing of SMS messaging data, timezones mean that data in particular languages may become simultaneously unavailable). Existing FL systems largely ignore such patterns, assuming node failures to be independent, leading to fragile training dynamics, stalled rounds, and poor fault tolerance. We begin by articulating and demonstrating these problems next~(\S\ref{s:problem}).

In response we propose \emph{AW-PSP}, a failure-resilient node selection mechanism that extends PSP with dynamic availability modeling and fairness-aware weighting. Unlike static and trace-based methods, AW-PSP predicts device availability in real time using lightweight Markov chains, allowing it to distinguish between transient and chronic failures, adapting sampling accordingly. It also reweights node selection based on both historical availability and recent utility, ensuring that even infrequently available but highly valuable participants are included. AW-PSP accounts for correlated failures by combining behavioral similarity (co-occurring availability patterns) with structural co-failure risks inferred from proximity metrics such as latency and packet loss. These signals are used to adjust per-node selection probabilities, actively diversifying participation to avoid synchronized dropouts. Finally, a Distributed Hash Table (DHT) overlay decentralizes metadata sharing, so that availability, latency, and freshness are propagated in a scalable and fault-tolerant manner, enabling proximity-aware and failure-resilient scheduling~(\S\ref{s:AW-PSP}).

We evaluate AW-PSP by isolating the impact of (un)availability on learning, using published trace data to model correlated and uncorrelated device unavailability coupled with ResNet-18/34 models applied to learn features of the CIFAR-10 dataset. We are not seeking to benchmark AW-PSP against the full landscape of client selection methods (Oort, FedGS, CA-Fed, etc) which often target faster convergence or throughput optimization under different assumptions. Rather, we consider how the additional mechanisms of AW-PSP improve performance over ``classic PSP'' in terms of fairness across nodes and classes, class coverage, and resilience~(\S\ref{s:evaluation}). We finish with a discussion of related work~(\S\ref{s:related}) and presentation of conclusions~(\S\ref{s:conclusions}).

Our contributions are threefold: (\emph{i})~we articulate and quantify the negative impact of availability unawareness in FL~(\S\ref{s:problem}); (\emph{ii})~we extend PSP with mechanisms that mitigate the effects of device unavailability~(\S\ref{s:AW-PSP}); and (\emph{iii})~we show both analytically and empirically that incorporating availability and failure awareness through these mechanisms improves fairness without sacrificing robustness~(\S\ref{s:evaluation}). PSP is a simple yet strong baseline for synchronization in distributed settings, and our results establish that AW-PSP preserves its lightweight and practical design while addressing previously overlooked issues of correlated failures, update starvation, and fairness imbalance. AW-PSP is thus a principled improvement over PSP that is in many ways orthogonal to rather than a universal replacement for existing FL client selection frameworks.

\section{The Problem}
\label{s:problem}

Previous works in the Fair FL literature have assumed that client availability is uncorrelated with
data distribution~\cite{mcmahan2017communication, wang2020optimizing}. This
assumption makes analysis easier but unfortunately fails to capture
real-world dynamics where correlations exist between certain types of data and
client availability. This leads to systematic under-representation of certain
groups/ classes, even when target weights are set to ensure proportionality of
contribution to the learning process with respect to the amount of data
available at each client.

We begin by demonstrating the practical impact of data-availability
co-correlation on system performance. Using Mininet~\cite{Yan2015VTMininet} we
emulate a small FL system comprising 10 clients, each exhibiting realistic device behavior such as intermittent availability due to charging, sleep cycles, or user activity. These clients simulate common patterns of correlated behavior observed in practice and enable us to examine how clients (groups) becoming unavailable at the same time can slow training progress and create fairness imbalances, regardless of whether the network itself is congested.

We examine how client unavailability affects model inference accuracy when
unavailability is both correlated and uncorrelated with features of the training
data distributed across clients. This enables detailed evaluation of how
different unavailability modes—mirroring real outages—impact federated learning
outcomes in a controlled, yet realistic, network environment.

Specifically, we fine-tune a ResNet34~\cite{He2016DeepResNet} model on a subset
of the CIFAR-10~\cite{Krizhevsky2009CIFAR10} dataset using federated training
over 50 epochs with Cross Entropy Loss~\cite{CrossEntropy} and the Adam
optimizer~\cite{Kingma2014Adam}. The central server broadcasts the current
global ResNet-34 model to all active workers, each of which performs $E$ local
training epochs on its class-specific data using cross-entropy loss and the Adam
optimizer. Workers then send their locally updated model weights back to the
server, which aggregates them by computing a weighted average. The global model
is updated with this averaged state, and the process repeats for a total of 50
communication rounds.

We assign data from specific CIFAR-10 classes evenly among workers,
$w_1, \ldots, w_N$, and train a ResNet34 DNN model from the resulting data
distributions according to FedAvg~\cite{mcmahan2017communication}. After training we simulate
an inference task by sending a mini-batch of images to each host, emulating
real-world edge deployments where inference is performed locally on distributed
devices after a federated training phase.

We then implemented a process for inducing correlated unavailability across
different groups of client devices. We started with the trace of devices involved
in an FL system published by ~\cite{yang2021characterizing}.
Based on the availability trace data provided in the GitHub link associated with
~\cite{yang2021characterizing}, the distribution of device availability percentages in Figure~\ref{f:distribution} indicates that the majority of devices exhibit less than 40\% availability, where device availability
percentage is defined as the percentage of time between the first and last times a device was seen to be live that it was available to perform FL, which presents challenges for designing reliable and robust distributed or FL systems over these edge devices. When high unavailability happens in a correlated manner, large groups of devices become simultaneously unavailable. This can drastically reduce the number of available clients, sometimes leaving too few (if any!) devices to participate.

This trace contains a
number of timestamped events per device, including \emph{WiFi on} and
\emph{battery charging on} which we take in combination to indicate that a
device is network-connected and plugged in (ready to participate in federated training according to Google~\cite{GoogleFL}), followed by \emph{WiFi
  off} or \emph{battery charged off} which we take to indicate
that the device has now gone to sleep, i.e.,~has become unavailable.

\begin{figure}
  \centering
  \includegraphics[width=1.0\linewidth]{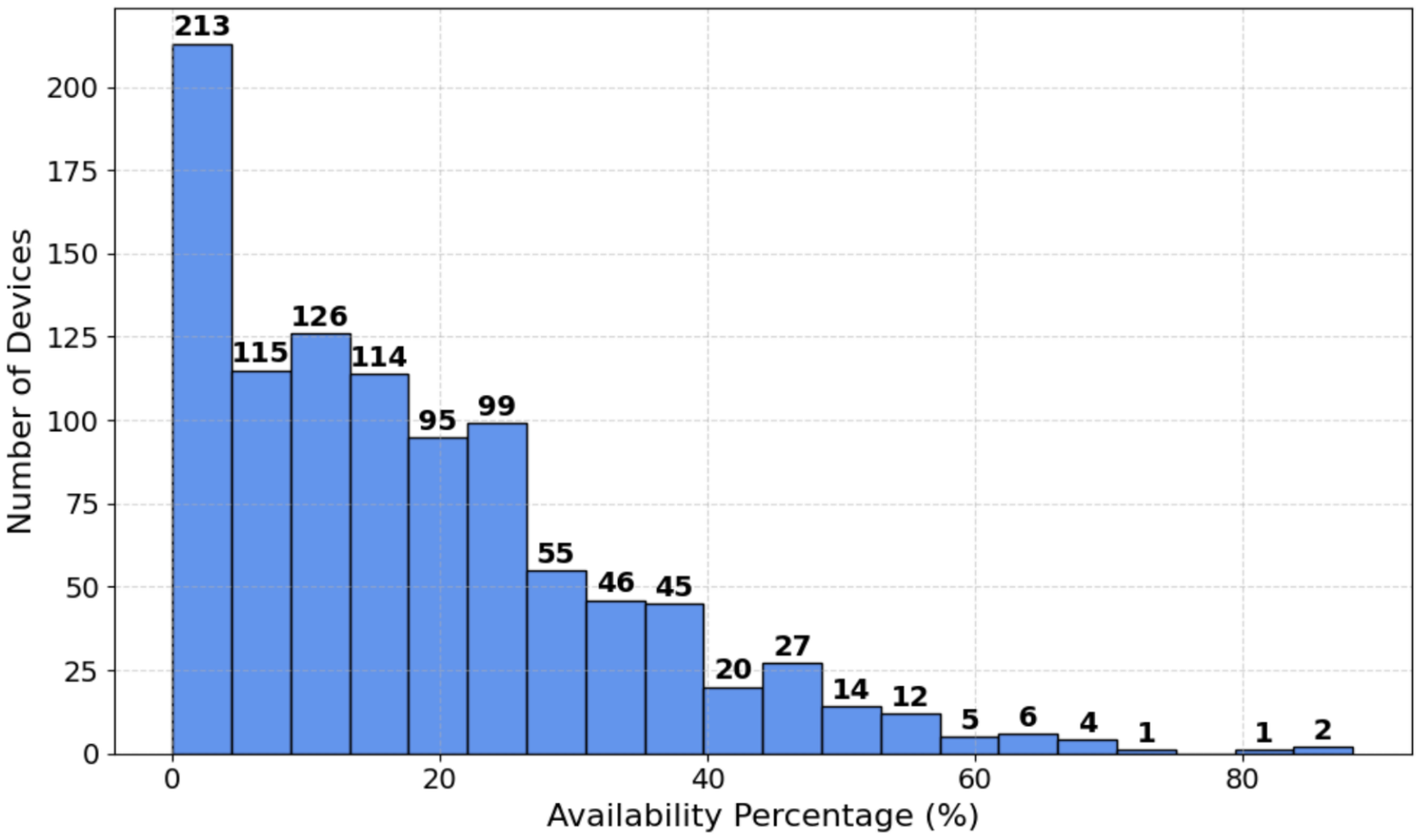}
  \caption{\label{f:distribution}Distribution of device availability percentages
    for the trace data~\cite{yang2021characterizing}, where \emph{device
      availability percentage} is defined as the percentage of time between the
    first and last times a device was seen to be live and available to
    perform FL, i.e.,~was charging and connected to Wi-Fi. Out of 1000 devices in
    the trace, 213 were available for $<$5\% of the time, and over 60\% were available for less than half the time.}
\end{figure}

We randomly sample from that trace to produce per-client device availability
traces, where device availability is correlated as described, over a specified
number of time steps. To compute the correlations between client availability,
we first generate binary availability vectors for each client with 1/0
indicating that the client is active/inactive in each time step. We then calculate
the pairwise correlation matrix from these vectors, quantifying how much clients
become unavailable and remain unavailable together. We also intermingle other
dropouts and reconnections at random to simulate behaviours such as signal loss,
device reboots, and app interruptions that would occur in practice.
This sampling process thus explores two different behaviours: \emph{(i) random}
where each node becomes unavailable independently with a fixed probability,
simulating uncorrelated unavailability, and \emph{(ii) correlated} where nodes fail
with probability dependent on the correlation of their unavailability with other
nodes -- so if one of two nodes that exhibit correlation above a threshold
becomes unavailable, so does the other -- simulating correlated unavailability.

The results (averaged over 5 retrials) show a clear structural difference between random and correlated unavailability. Under random failures (Table~\ref{tab:random_summary_merged}), each host fails independently with probability $p_f$, producing a stochastic spread of unavailable nodes. As $p_f$ increases, the number of active hosts decreases gradually. Performance exhibits distinct plateaus determined by the number of labels per client (approximately 73\%, 76\%, and 78.5\% for labels/client = 1, 5, and 10, respectively). This indicates that, under random failures, the reported mean accuracy is largely governed by the underlying data partitioning—specifically label coverage—rather than the exact number of active participants, provided that sufficient clients remain to cover the label space.
This stability is partly a consequence of the evaluation protocol: accuracy is computed only over labels that remain covered by at least one active host, and missing labels are excluded rather than counted as errors. When the number of active clients drops below this implicit coverage threshold, large portions of the label space become unrepresented, leading to sharp degradation in (or even undefined) accuracy.

Under correlated failures (Table~\ref{tab:correlated_summary_merged}), unavailability is clustered in that hosts fail together if they share high availability correlation, defined by a threshold $c$. For low correlation thresholds ($c \le 0.4$), all hosts can become unavailable simultaneously, resulting in complete system collapse and zero accuracy. Once the correlation threshold exceeds that threshold ($c \ge 0.5$), the system transitions sharply from minimal participation (1–2 active hosts) to a stable operating regime with sufficient active clients. Notably, the mean accuracy immediately recovers to plateau levels comparable to random failures, showing that co-correlation primarily affects availability, not model quality. This demonstrates that correlated failures introduce abrupt, threshold-driven system instability that is not captured by independent failure models.

The per-class results (Tables~\ref{tab:random_perclass_merged} and \ref{tab:correlated_perclass_merged}) further highlight this distinction. Under random failures, class coverage degrades in a scattered manner, with different subsets of labels missing across runs. In contrast, correlated failures produce structured coverage gaps, where specific groups of classes disappear together due to shared availability patterns. When only a small number of hosts remain active, the model is effectively trained on a reduced label space, leading to biased predictions. This effect is most pronounced for labels/client = 1, where each client contributes a narrow slice of the data distribution and correlated unavailability can eliminate entire regions of the label space. For higher label densities (5 and 10), although individual clients hold more diverse data, correlated failures can still remove overlapping subsets of labels simultaneously. As a result, redundancy does not fully prevent coverage loss; instead, it mainly stabilizes performance once a sufficient subset of classes remains available, preserving the plateau behavior observed in the aggregate results.

Overall, the number of active hosts determines whether the system is operational, while the joint structure of data partitioning and availability correlation determines the effective training distribution and thus the achievable accuracy. Random failures primarily reduce participation, whereas correlated failures reshape the data distribution itself by removing coherent subsets of clients. Increasing labels per client improves robustness by introducing redundancy in class coverage, making the system resilient to both forms of unavailability. Conversely, sparse partitions amplify the impact of correlated failures, leading to threshold-driven collapse, systematic bias, and reduced generalization. These findings (relevant to non-IID effects and data heterogeneity) align with prior work on mitigating non-IID effects through data sharing, personalized FL, or heterogeneity-aware aggregation techniques~\cite{zhao2018federated, li2019fedmd, fallah2020personalized, smith2017federated, li2020federated, li2019fair, sattler2020clustered}, highlighting the critical interplay between failure dynamics and data heterogeneity in shaping system behavior.

\begin{table*}[ht]
\centering
\footnotesize
\caption{Summary results under random failures for different numbers of labels per client.}
  \vspace{-.3cm}
\label{tab:random_summary_merged}
\begin{tabular}{c|ccc|ccc|ccc}
\hline
& \multicolumn{3}{c|}{Labels/client = 1} & \multicolumn{3}{c|}{Labels/client = 5} & \multicolumn{3}{c}{Labels/client = 10} \\
$p_f$ & Mean Acc. & Active & Failed & Mean Acc. & Active & Failed & Mean Acc. & Active & Failed \\
\hline
0.1 & 73.50 & 10 & 0 & 76.00 & 9 & 1 & 78.50 & 10 & 0 \\
0.2 & 73.50 & 10 & 0 & 76.00 & 9 & 1 & 78.50 & 7 & 3 \\
0.3 & 75.63 & 8  & 2 & 76.00 & 9 & 1 & 78.50 & 3 & 7 \\
0.4 & 67.50 & 6  & 4 & 76.00 & 7 & 3 & 78.50 & 4 & 6 \\
0.5 & 78.33 & 3  & 7 & 76.00 & 5 & 5 & 78.50 & 6 & 4 \\
0.6 & 65.00 & 5  & 5 & 76.00 & 3 & 7 & 78.50 & 6 & 4 \\
0.7 & 76.25 & 4  & 6 & 76.00 & 2 & 8 & 78.50 & 1 & 9 \\
0.8 & 85.00 & 1  & 9 & --    & 0 & 10 & 78.50 & 2 & 8 \\
0.9 & --    & 0  & 10 & 76.00 & 2 & 8 & --    & 0 & 10 \\
\hline
\end{tabular}%
\end{table*}

\begin{table*}[ht]
\centering
\footnotesize
\caption{Summary results under correlated failures for different numbers of labels per client.}
  \vspace{-.3cm}
\label{tab:correlated_summary_merged}
\begin{tabular}{c|ccc|ccc|ccc}
\hline
& \multicolumn{3}{c|}{Labels/client = 1} & \multicolumn{3}{c|}{Labels/client = 5} & \multicolumn{3}{c}{Labels/client = 10} \\
$c$ & Mean Acc. & Active & Failed & Mean Acc. & Active & Failed & Mean Acc. & Active & Failed \\
\hline
0.1 & --    & 0  & 10 & --    & 0  & 10 & --    & 0  & 10 \\
0.2 & --    & 0  & 10 & --    & 0  & 10 & --    & 0  & 10 \\
0.3 & --    & 0  & 10 & --    & 0  & 10 & --    & 0  & 10 \\
0.4 & --    & 0  & 10 & --    & 0  & 10 & --    & 0  & 10 \\
0.5 & 70.00 & 1  & 9  & 83.00 & 1  & 9  & 78.50 & 1  & 9  \\
0.6 & 60.00 & 2  & 8  & 83.00 & 2  & 8  & 78.50 & 2  & 8  \\
0.7 & 66.67 & 6  & 4  & 76.00 & 6  & 4  & 78.50 & 6  & 4  \\
0.8 & 73.50 & 10 & 0  & 76.00 & 10 & 0  & 78.50 & 10 & 0  \\
0.9 & 73.50 & 10 & 0  & 76.00 & 10 & 0  & 78.50 & 10 & 0  \\
\hline
\end{tabular}%
\end{table*}

\setlength{\tabcolsep}{4pt}  
\begin{table*}[ht]
\centering
\footnotesize
\caption{Per-class accuracies under random failures for labels/client = 1, 5, and 10. Empty cells indicate no value could be computed because no response was received from a host trained on that class.}
  \vspace{-.3cm}
\label{tab:random_perclass_merged}
\begin{tabular}{c|cccccccccc|cccccccccc|cccccccccc}
\hline
& \multicolumn{10}{c|}{Labels/client = 1} & \multicolumn{10}{c|}{Labels/client = 5} & \multicolumn{10}{c}{Labels/client = 10} \\
$p_f$ & 1 & 2 & 3 & 4 & 5 & 6 & 7 & 8 & 9 & 10
     & 1 & 2 & 3 & 4 & 5 & 6 & 7 & 8 & 9 & 10
     & 1 & 2 & 3 & 4 & 5 & 6 & 7 & 8 & 9 & 10 \\
\hline
0.1 & 65 & 90 & 50 & 85 & 70 & 80 & 50 & 75 & 85 & 85
    & 56 & 85 & 80 & 90 & 72 & 80 & 76 & 80 & 72 &
    & 90 & 80 & 85 & 80 & 60 & 65 & 70 & 85 & 95 & 75 \\

0.2 & 65 & 90 & 50 & 85 & 70 & 80 & 50 & 75 & 85 & 85
    & 56 & 85 & 80 & 72 & 90 & 76 & 80 & 80 & 72 &
    & 70 & 77 & 83 & 93 & 80 & 77 & 65 &  &  &  \\

0.3 & 65 & 90 & 85 & 70 & 50 & 75 & 85 & 85 &  &
    & 85 & 76 & 92 & 75 & 76 & 80 & 55 & 88 & 50 &
    & 77 & 76 & 83 &  &  &  &  &  &  &  \\

0.4 & 65 & 50 & 70 & 50 & 85 & 85 &  &  &  &
    & 56 & 80 & 72 & 91 & 76 & 77 & 72 &  &  &
    & 84 & 82 & 76 & 72 &  &  &  &  &  &  \\

0.5 & 90 & 70 & 75 &  &  &  &  &  &  &
    & 66 & 80 & 91 & 72 & 77 &  &  &  &  &
    & 80 & 73 & 80 & 73 & 80 & 87 &  &  &  &  \\

0.6 & 65 & 50 & 50 & 75 & 85 &  &  &  &  &
    & 66 & 72 & 83 &  &  &  &  &  &  &
    & 80 & 73 & 80 & 73 & 80 & 87 &  &  &  &  \\

0.7 & 90 & 50 & 85 & 80 &  &  &  &  &  &
    & 83 & 69 &  &  &  &  &  &  &  &
    & 79 &  &  &  &  &  &  &  &  &  \\

0.8 & 85 &  &  &  &  &  &  &  &  &
    &  &  &  &  &  &  &  &  &  &
    & 80 & 77 &  &  &  &  &  &  &  &  \\

0.9 &  &  &  &  &  &  &  &  &  &
    & 83 & 69 &  &  &  &  &  &  &  &
    &  &  &  &  &  &  &  &  &  &  \\
\hline
\end{tabular}%
\end{table*}

\setlength{\tabcolsep}{4pt}  
\begin{table*}[ht]
\centering
\footnotesize
\caption{Per-class accuracies under correlated failures for labels/client = 1, 5, and 10. Empty cells indicate no value could be computed because no response was received from a host trained on that class.}
  \vspace{-.3cm}
\label{tab:correlated_perclass_merged}
\begin{tabular}{c|cccccccccc|cccccccccc|cccccccccc}
\hline
& \multicolumn{10}{c|}{Labels/client = 1} & \multicolumn{10}{c|}{Labels/client = 5} & \multicolumn{10}{c}{Labels/client = 10} \\
$c$  & 1 & 2 & 3 & 4 & 5 & 6 & 7 & 8 & 9 & 10
     & 1 & 2 & 3 & 4 & 5 & 6 & 7 & 8 & 9 & 10
     & 1 & 2 & 3 & 4 & 5 & 6 & 7 & 8 & 9 & 10 \\
\hline
0.1 &  &  &  &  &  &  &  &  &  &
    &  &  &  &  &  &  &  &  &  &
    &  &  &  &  &  &  &  &  &  &  \\

0.2 &  &  &  &  &  &  &  &  &  &
    &  &  &  &  &  &  &  &  &  &
    &  &  &  &  &  &  &  &  &  &  \\

0.3 &  &  &  &  &  &  &  &  &  &
    &  &  &  &  &  &  &  &  &  &
    &  &  &  &  &  &  &  &  &  &  \\

0.4 &  &  &  &  &  &  &  &  &  &
    &  &  &  &  &  &  &  &  &  &
    &  &  &  &  &  &  &  &  &  &  \\

0.5 & 70 &  &  &  &  &  &  &  &  &
    & 83 &  &  &  &  &  &  &  &  &
    & 79 &  &  &  &  &  &  &  &  &  \\

0.6 & 50 & 70 &  &  &  &  &  &  &  &
    & 76 & 90 &  &  &  &  &  &  &  &
    & 80 & 77 &  &  &  &  &  &  &  &  \\

0.7 & 65 & 90 & 50 & 70 & 50 & 75 &  &  &  &
    & 66 & 76 & 92 & 76 & 88 & 72 &  &  &  &
    & 80 & 73 & 80 & 73 & 80 & 87 &  &  &  &  \\

0.8 & 65 & 90 & 50 & 85 & 70 & 80 & 50 & 75 & 85 & 85
    & 85 & 85 & 80 & 75 & 90 & 80 & 80 & 55 & 80 & 50
    & 90 & 80 & 85 & 80 & 60 & 65 & 70 & 85 & 95 & 75 \\

0.9 & 65 & 90 & 50 & 85 & 70 & 80 & 50 & 75 & 85 & 85
    & 85 & 85 & 80 & 75 & 90 & 80 & 80 & 55 & 80 & 50
    & 90 & 80 & 85 & 80 & 60 & 65 & 70 & 85 & 95 & 75 \\
\hline
\end{tabular}%
\end{table*}

\section{Availability Weighted PSP}
\label{s:AW-PSP}

In AW-PSP, the probability of selecting a node is no longer uniform as in Classic-PSP, but instead depends on a composite availability score that incorporates three factors: the node’s ability to complete both computation and communication within a tolerable deadline, its likelihood of recovering from recent failures, and its risk of being affected by correlated failures with nearby nodes. Concretely, each node’s selection probability combines its empirical availability (estimated through recent success rates of computation and communication), a recovery probability that accounts for transient failures and allows temporarily offline nodes to rejoin, and a group-level correlation penalty that reduces the weight of nodes likely to fail together (e.g., on the same subnet or switch). These availability estimates are further refined by blending local history with information from neighboring nodes in the DHT overlay, ensuring that proximity in terms of latency and computation speed also informs the sampling. By embedding these enhanced availability scores into the PSP framework, AW-PSP adaptively prioritizes nodes that are both reliable and diverse, thereby mitigating stragglers, reducing synchronization delays, and improving fairness compared to the purely random selection of Classic-PSP.

\subsection{Modeling Node Availability}
Node availability can be defined as:

\begin{equation}
a_i(t) = P(T_{i}^{\text{comp}}(t) + T_{i}^{\text{comm}}(t) \leq T_{\max})
\end{equation}
where:
\begin{itemize}
    \item $T_{i}^{\text{comp}}(t)$ is the computation time of node $i$ for the current training round.
    \item $T_{i}^{\text{comm}}(t)$ is the communication delay for sending the model update.
    \item $T_{\max}$ is the maximum tolerable delay before a node is considered unavailable.
\end{itemize}
A node is available if it can complete and communicate updates within $T_{\max}$ with high probability.
We estimate $a_i(t)$ using historical response times with an exponentially weighted moving average (EWMA):

\begin{equation}
a_i(t) = \lambda a_i(t-1) + (1 - \lambda) I(T_{i}^{\text{comp}}(t) + T_{i}^{\text{comm}}(t) \leq T_{\max})
\end{equation}

where, $I(\cdot)$ is an indicator function. $\lambda$ (e.g., 0.9) controls the decay rate, giving more or less weight to recent behavior.

\begin{equation}
I(T_{i}^{\text{comp}}(t) + T_{i}^{\text{comm}}(t) \leq T_{\max}) =
\begin{cases}
1, & \text{if node } i \text{ finished on time} \\
0, & \text{otherwise}
\end{cases}
\end{equation}

\textbf{Predicting node availability based on DHT:} To predict node availability in computation and communication, we leverage Distributed Hash Table (DHT)-based metadata to benefit from Physical and Computation Proximity:

\begin{enumerate}
    \item \textbf{Physical proximity:}
    \begin{equation}
        d_{i,j} = RTT_{i,j} = RTT_{j,i}
    \end{equation}
    \begin{equation}
        d_{i,j} \leq \tau_d
    \end{equation}

    \item \textbf{Computation proximity:}
    \begin{equation}
        \Delta T_{i,j}^{\text{comp}} = |T_i^{\text{comp}} - T_j^{\text{comp}}|
    \end{equation}
    \begin{equation}
        \Delta T_{i,j}^{\text{comp}} \leq \tau_c
    \end{equation}
\end{enumerate}

 We estimate computation and communication availability using recent history across the last $T$ rounds:

\begin{equation}
    a_i^{\text{comp}} = \frac{1}{T} \sum_{t=1}^{T} H_i^{\text{comp}}(t),
    \quad
    a_i^{\text{comm}} = \frac{1}{T} \sum_{t=1}^{T} H_i^{\text{comm}}(t)
\end{equation}

\begin{itemize}
    \item $H_i^{\text{comp}}(t) = 1$ if the node has participated in more than a threshold fraction of rounds ($70\%$) up to time $t$, and $0$ otherwise. $\tau_c$ is a threshold for historical success rate.
    \item $H_i^{\text{comm}}(t) = 1$ if latency is below a threshold $\tau_{\text{lat}}$ (e.g., $100$\,ms) and packet loss below a threshold $\tau_{\text{loss}}$ (e.g., $40\%$), and $0$ otherwise.
\end{itemize}

The overall availability is then:

\begin{equation}
    a_i(t) = a_i^{\text{comp}} \cdot a_i^{\text{comm}}
\end{equation}

\textbf{Availability based on recovery probability:}
The recovery probability, \(\beta_i\), represents the likelihood that a node will recover from a failure in the current round and return to completing its tasks within the maximum allowable time. The recovery probability \(\beta_i\) is modeled as the conditional probability that the node will meet the deadline in the next round, given that it failed in the previous round.
We define \(\beta_i\) as:
\begin{equation}
\begin{aligned}
\beta_i(t) \;&=\; \Pr\big(T_i^{\mathrm{comp}}(t+1) + T_i^{\mathrm{comm}}(t+1) \le T_{\max} \;\big|\;\\
&T_i^{\mathrm{comp}}(t) + T_i^{\mathrm{comm}}(t) > T_{\max}\big)
\end{aligned}
\end{equation}
\[
\beta_i(t) \approx \frac{\text{\# times node $i$ recovered  after a failure (in last $T$ rounds)}}{\text{\# times node $i$ failed in last $T$ rounds}}
\]
which takes into account the node's previous failures and provides a more nuanced prediction for future availability. This approach assumes that there is some likelihood that a node may recover from a failure, possibly due to transient network issues, resource availability changes, or system optimization.
Future availability is estimated based on the node’s previous availability and its predicted return probability:

\begin{equation}
    \tilde{a}_i(t) = \tilde{a}_i(t-1) + (1 - \tilde{a}_i(t-1)) \beta_i(t)
\end{equation}

where \(\tilde{a}_i(t)\) is the baseline availability estimate.

\subsection{Modeling Runtime Co-Correlation in Federated Edge Environments}

In real-world federated systems, device failures and unavailabilities are often correlated rather than independent. Devices co-located in the network topology (e.g., sharing the same subnet or region) or exhibiting similar usage patterns (e.g., nighttime disconnection) may fail together.

\paragraph{Trace-based correlation ($\gamma^{\text{trace}}_{i,j}$).}
This denotes a static correlation score between nodes $i$ and $j$, computed offline from their historical binary availability traces. It is loaded before the simulation and remains fixed during execution. This score captures long-term similarity in availability behavior.

\paragraph{Runtime co-failure correlation ($\gamma^{\text{fail}}_{i,j}(t)$).}
This denotes a dynamic co-failure score between nodes $i$ and $j$ at round $t$, computed from the history of rounds in which both nodes were observed as failed. In the implementation, each pair of failed neighboring nodes records the current round in a \texttt{failure\_correlation} structure, and the runtime score is defined as the fraction of elapsed rounds in which the pair has failed together. This score is updated online during execution and captures recent correlated-failure behavior.

\begin{equation}
\gamma_{i,j}(t) = \alpha \cdot \gamma^{\text{trace}}_{i,j} + (1 - \alpha) \cdot \gamma^{\text{fail}}_{i,j}(t)
\end{equation}

where \(\alpha \in [0,1]\) balances the contribution between historical co-behavior and recent co-failure evidence.
Each node \(i\) is associated with a set of significantly correlated peers:
\(\mathcal{G}_i(t) = \{ j \mid \gamma_{i,j}(t) > \tau_{\text{corr}} \}\)
and its correlation penalty as:

\begin{equation}
\rho_i(t) = \sum_{j \in \mathcal{G}_i(t)} \hat{C}_{i,j} \cdot \gamma_{i,j}(t)
\end{equation}

where \(\hat{C}_{i,j}\) is a proximity weight, e.g., based on DHT logical distance or latency measurements. We define a latency threshold in the empirical section. The correlation penalty adjusts the availability prediction for node \(i\) as follows:

\begin{equation}
a_i^{\text{adj}}(t) = \tilde{a}_i(t) \cdot (1 - \rho_i(t))
\end{equation}

 This formulation downweights nodes that are heavily entangled in unreliable clusters, improving resilience against cascading or simultaneous failures.

\subsection{Sampling Probability in AW-PSP}

We define the final sampling probability \( p_i' \) for each node \(i\) below. The theoretical justification and guarantees aree discussed in appendix A.

\begin{equation}
p_i'(t) = p \cdot \left[ \left(a_i^{\text{comp}} \cdot a_i^{\text{comm}}\right) + \left(1 - a_i^{\text{comp}} \cdot a_i^{\text{comm}}\right) \cdot \beta_i(t) \right] \cdot \left(1 - \rho_i(t)\right)
\end{equation}

where:
\begin{itemize}
  \item \( p \) is the baseline uniform sampling probability,
  \item \( a_i^{\text{comp}} \cdot a_i^{\text{comm}} \) is the immediate availability estimate of node \(i\),
  \item \( \beta_i(t) \) is the return (recovery) probability if the node was recently unavailable,
  \item \( \rho_i(t) \) is the correlation penalty capturing the node's entanglement with unstable peers (frequently failing, dropping, or unreliable).
\end{itemize}

\section{Evaluation}
\label{s:evaluation}

Our methodological goal is to isolate the effects of availability-weighted sampling and correlated failures on distributed synchronization. Below we quantify the workload and show that a single, well-provisioned workstation suffices while preserving experimental control and reproducibility.

\textbf{Workload characteristics.}
CIFAR-10 occupies $\sim$170\,MB on disk; when cached as \texttt{float32}, it requires $\sim$220\,MB RAM if fully resident, but our experiments stream in mini-batches or single images during inference, keeping the steady-state working set substantially lower. We use ResNet-18/34 for inference on containerized clients. ResNet-18 has $\sim$11.7M parameters ($\sim$45\,MB weights), ResNet-34 $\sim$21.8M ($\sim$85\,MB weights). With batch size~32 and no gradient storage, per-process inference memory is typically $\sim$150--300\,MB (weights + activations + framework/runtime). We start with 10 concurrent physical clients plus one coordinator.

\textbf{Training and Resource implications.}
Training increases per-process memory and CPU demands due to activations, optimizer state, and gradient buffers. We use a single server host with 64 GB RAM and 110 vCPUs (our bottleneck); also we adopt practical mitigations: (\emph{i})~\textit{small local batch sizes of 32 images}, (\emph{ii})~modest model variant (ResNet-18) which emulates a mobile device resource capacity, and (\emph{iii})~\textit{Docker resource limits} to emulate heterogeneous device capacities. In reality, each client utilizes up to 800\% CPU (Saturating all 8 cores) with an average memory usage of 687 MB.

Each global round consists of (\emph{i})~selecting a subset of clients according to the sampling policy (AW-PSP or Classic-PSP), (\emph{ii})~dispatching the current global model to the selected clients, (\emph{iii})~each selected client performing a fixed number of \emph{local} training epochs (three local epochs in our setup), and (\emph{iv})~returning the local update to the coordinator for aggregation. We run the full experiment for 50 global epochs (more epochs result in no performance improvement). After each global epoch the coordinator evaluates the aggregated global model by distributing a small test partition to available clients and collecting local predictions; these per-epoch evaluations produce the accuracy/loss curves reported in the results. In our experimental setup, we fixed the target probabilistic client selection to 5 nodes per round, while assuming 4 neighbors per node in the DHT overlay.

\textbf{Process isolation.}
Each client and the coordinator run in separate Docker containers. In our setup we instantiate 11 containers (10 clients and 1 coordinator). Each container runs an independent Python runtime with the model weights and inference scripts. The CIFAR-10 dataset is mounted read-only from the host into each container, which avoids duplication of storage while ensuring consistent data availability across all clients.FL deployments in real-world environments often involve thousands to millions of devices, such as mobile phones and edge nodes, which exhibit heterogeneous data distributions, intermittent availability, and correlated failures. However, it is impractical to physically deploy and execute training across such a large number of devices in a controlled experimental environment. Even with cloud infrastructure, running hundreds or thousands of independent training processes simultaneously introduces significant overhead in compute, memory, and orchestration.
To bridge this gap, we adopt the concept of \emph{logical clients}. Logical clients represent a large population of virtual participants whose availability patterns, data characteristics, and correlation structures are explicitly modeled, while only a smaller number of physical workers execute training. This abstraction enables us to evaluate FL algorithms under realistic large-scale conditions.
In each training round, the server first selects a subset of logical clients according to the scheduling policy. These selected logical clients are then mapped to a smaller pool of physical clients using a wave-based execution mechanism. If $m=[10,30,100]$ logical clients are selected among $N==[100,300,1000,3000]$ population and $P$ physical clients are available, training is executed in $\lceil m / P \rceil$ waves, where each physical client sequentially processes multiple logical clients. This decoupling enables scalable evaluation while maintaining bounded system resource usage.

\textbf{Failure injection.}
Failures are injected at the network layer using the \texttt{TopologyProvider} abstraction, which wraps each client container with configurable link properties. For each node, we instantiate e.g. \texttt{link\_latency}=20, \texttt{link\_loss}=5. Here, \texttt{link\_latency} and \texttt{link\_loss} specify the baseline delay (20\,ms) and packet loss rate (5\%) applied to the virtual link connecting the container to the bridge. These impairments are implemented internally with Linux, ensuring they are enforced at the veth interface that connects the container to the virtual switch.
To detect failures, we continuously measure the round-trip latency of each client by issuing controlled \texttt{ping} probes. A node is considered \emph{failed} if the measured latency exceeds a threshold of 100\,ms, which indicates either severe congestion or effective disconnection. The choice of 100\,ms is motivated by synchronization requirements in distributed training: latencies above this range cause straggler effects that significantly degrade training throughput and convergence speed. Thus, while a node may still be technically reachable, any client exceeding this bound is treated as unavailable for synchronization purposes.

\subsection{Trace Data and Correlated Failures}
\label{s:corrfailures}
While trace data provides insight into device availability patterns, it does not capture actual runtime failures that occur during training—such as dropped connections or infrastructure-related outages. To model better these conditions, we introduce synthetic failures driven by real-time telemetry:

\begin{itemize}
\item \textbf{Independent Failures:} A node is marked as failed during a round if its measured latency exceeds a defined threshold (here we consider 100ms) or if its packet loss rate indicates unreliable connectivity. These failures are dynamically detected based on runtime metrics rather than injected probabilistically.

\item \textbf{Correlated Failures:} Nodes that are topologically or behaviorally close (e.g., sharing DHT proximity or latency profiles, where 4 closest nodes are assumed neighbors) are monitored for correlated failure likelihood. If a node fails, its neighbor with high correlation score is considered as a correlated failure, if trace score or fail score is greater than a threshold, simulating clustered failures due to shared network infrastructure.
\end{itemize}

These failure events are reflected in each node's DHT metadata, which includes latency measurements, packet loss, recent availability history, and round participation. The AW-PSP strategy leverages this decentralized availability data to prioritize reliable nodes and avoid synchronized dropouts.
Each node periodically updates its availability based on its recent computation and communication success rates. Availability is calculated separately for computation and communication tasks and integrated into the metadata for querying, according to Eq. 8,
where $H_i^{\text{comp}}(t)$ and $H_i^{\text{comm}}(t)$ indicate success (1) or failure (0) of computation and communication tasks at time $t$, respectively. If latency is greater than a threshold or server state update exceeds a timeout, the node is considered failed.
Each node periodically updates its metadata in the DHT based on its availability and performance.

\begin{algorithm}[ht]
\scriptsize
\caption{Federated Learning with DHT-Based Node Selection}
\begin{algorithmic}[1]
    \FOR{each candidate node $i$}
        \STATE \textbf{Availability:}
        \[
    \tilde{a}_i(t) = \tilde{a}_i(t-1) + (1 - \tilde{a}_i(t-1)) \beta_i(t)
        \]

        \STATE \textbf{Correlated Failure Risk:}
        \[
        \rho_i(t) = \sum_{j \in \mathcal{G}_i} \hat{C}_{i,j} \cdot \gamma_{i,j}(t)
        \]

        \STATE \textbf{Adjusted Availability:}
        \[
        a_i^{\text{adj}}(t) = \tilde{a}_i(t) \cdot (1 - \rho_i(t))
        \]

         \STATE \textbf{Final Sampling Probability:}
        \[
         \begin{aligned}
        p_i'(t) &= p \cdot \Big[\,(a_i^{\text{comp}} \cdot a_i^{\text{comm}}) + (1 - a_i^{\text{comp}} \cdot a_i^{\text{comm}})
        \cdot \beta_i(t)\,\Big] \\
        & \quad\quad \cdot (1 - \rho_i(t))
        \end{aligned}
        \]
  \
\ENDFOR
  \STATE \textbf{Node Selection:} Choose top $N$ nodes with highest $\text{score}_i$.

    \STATE \textbf{Distribute Training:} Send global model and instructions to selected nodes.

    \STATE \textbf{Local Training:} Each node trains on its local data and sends model updates back.

    \STATE \textbf{Aggregation:} Server aggregates all updates and refines the global model.

\end{algorithmic}
\end{algorithm}

To evaluate the effectiveness of our proposed sampling strategy, we implement a baseline version of PSP (Classic-PSP) that selects a fixed number of nodes uniformly at random from the pool of currently online devices, without considering any availability, freshness (i.e.,~recency of participation), or failure correlation information. This serves as a control to isolate the benefits introduced by the AW-PSP mechanism.
Our AW-PSP algorithm prioritizes node selection based on a dynamic score that integrates real-time availability predictions, freshness, and class diversity, according to Algorithm 1. Additionally, it excludes nodes participating in correlated failure clusters, identified using two orthogonal signals: user behavior traces and proximity-based co-failure detection (per~\S\ref{s:corrfailures}). This proactive filtering ensures robustness against subnet-level outages or synchronous dropouts.
The core selection mechanism for the AW-PSP strategy aims to choose a subset of active nodes to participate in the current training round. This selection is done by prioritizing nodes based on their predicted availability, their freshness, and ensuring class label coverage:

\begin{itemize}
\item Filter out unavailable or failed nodes, including those involved in correlated failures.
\item Among the remaining active nodes, calculate a selection score for each defined as the product of that node’s availability and freshness. Nodes that are highly available and recently active are ranked higher.
\item From the top-ranked nodes,  greedily select clients to maximize class label coverage, ensuring that the selected nodes represent as many data labels as possible.
\item For each selected node, evaluate their contribution by computing loss deltas.
\item Finally,  calculate three fairness metrics by analyzing the variance of loss deltas across nodes and classes.
\end{itemize}

\subsection{Results and Discussion}

In designing our evaluation, we explicitly structure the experiments to answer core research questions:

\begin{itemize}
    \item \textbf{RQ1:} How does client-level data heterogeneity impact global accuracy?
    \item \textbf{RQ2:} Does AW-PSP improve fairness across nodes and classes compared to Classic-PSP?
    \item \textbf{RQ3:} How well does AW-PSP maintain class coverage and balance label participation relative to Classic-PSP?
    \item \textbf{RQ4:} Under correlated failures, does AW-PSP reduce the fraction of failed selections and improve resilience?
\end{itemize}

\subsubsection{Global accuracy for different heterogeneity}
In order to evaluate the impact of data heterogeneity on system performance, we calculate the overall training accuracy under different label distributions across clients in a system with 100 clients. Specifically, we consider increasingly less heterogeneous settings in which each client holds 2 labels, 5 labels, or all 10 labels. As shown in Figure~\ref{fig:accuracy-comparison}, the results indicate that AW-PSP consistently outperforms both PSP and Oort across all heterogeneity levels. In the 2-label setting, where heterogeneity is strongest, all methods exhibit substantial fluctuations and relatively low accuracy, but AW-PSP generally attains higher peaks and stronger late-round performance than PSP and Oort, showing better robustness under severe non-IID conditions. In the 5-label setting, although the three methods remain unstable, AW-PSP still tends to achieve competitive or superior accuracy more frequently than the baselines, again remaining above PSP and Oort in many rounds. In the 10-label setting, where data distribution is more balanced, all methods converge more smoothly near 72\% accuracy, but AW-PSP remains slightly ahead overall, demonstrating the best and most stable final performance among the three approaches.
Overall, these results show two important findings. First, increasing the number of labels per client reduces the non-IID effect, leading to higher accuracy and more stable convergence. Second, and more importantly, AW-PSP consistently surpasses PSP and Oort under all tested label-distribution settings. Oort is implemented by  vendoring the \texttt{oort.py} module from \url{https://github.com/SymbioticLab/Oort} for the server to import and run that selector directly in the training pipeline for direct comparison against AW-PSP and Classic-PSP. This implementation includes Oort exploration/exploitation, pacer logic, blacklist/duration penalties, while the setup adds utility wiring and logs fairness/accuracy metrics.

\begin{figure*}[ht]
    \centering
    \subfigure[Accuracy test for 2 Labels/Client]{\includegraphics[width=0.27\textwidth]{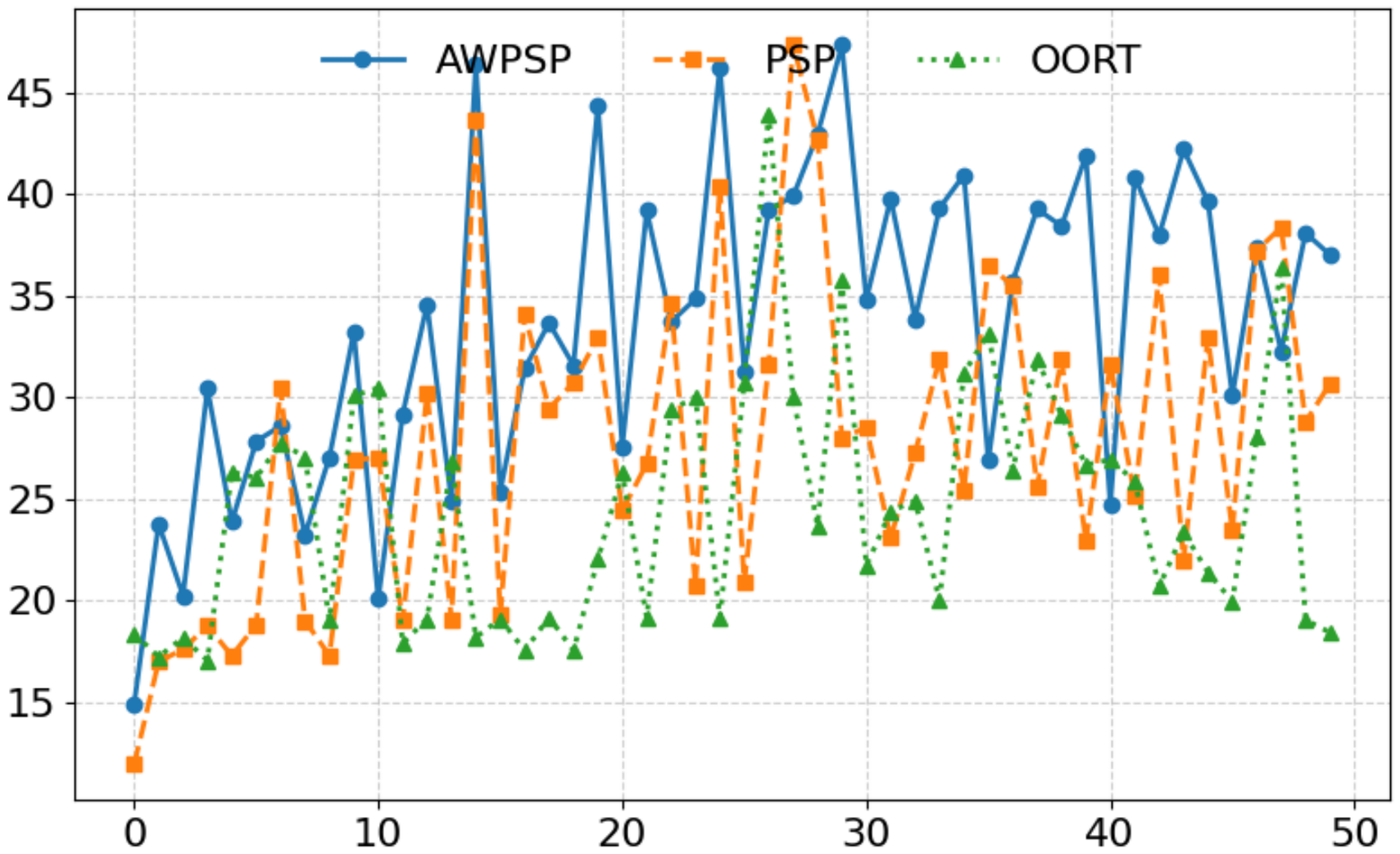}\label{accuracy-2labels}}
    \subfigure[Accuracy test for 5 Labels/Client]{\includegraphics[width=0.27\textwidth]{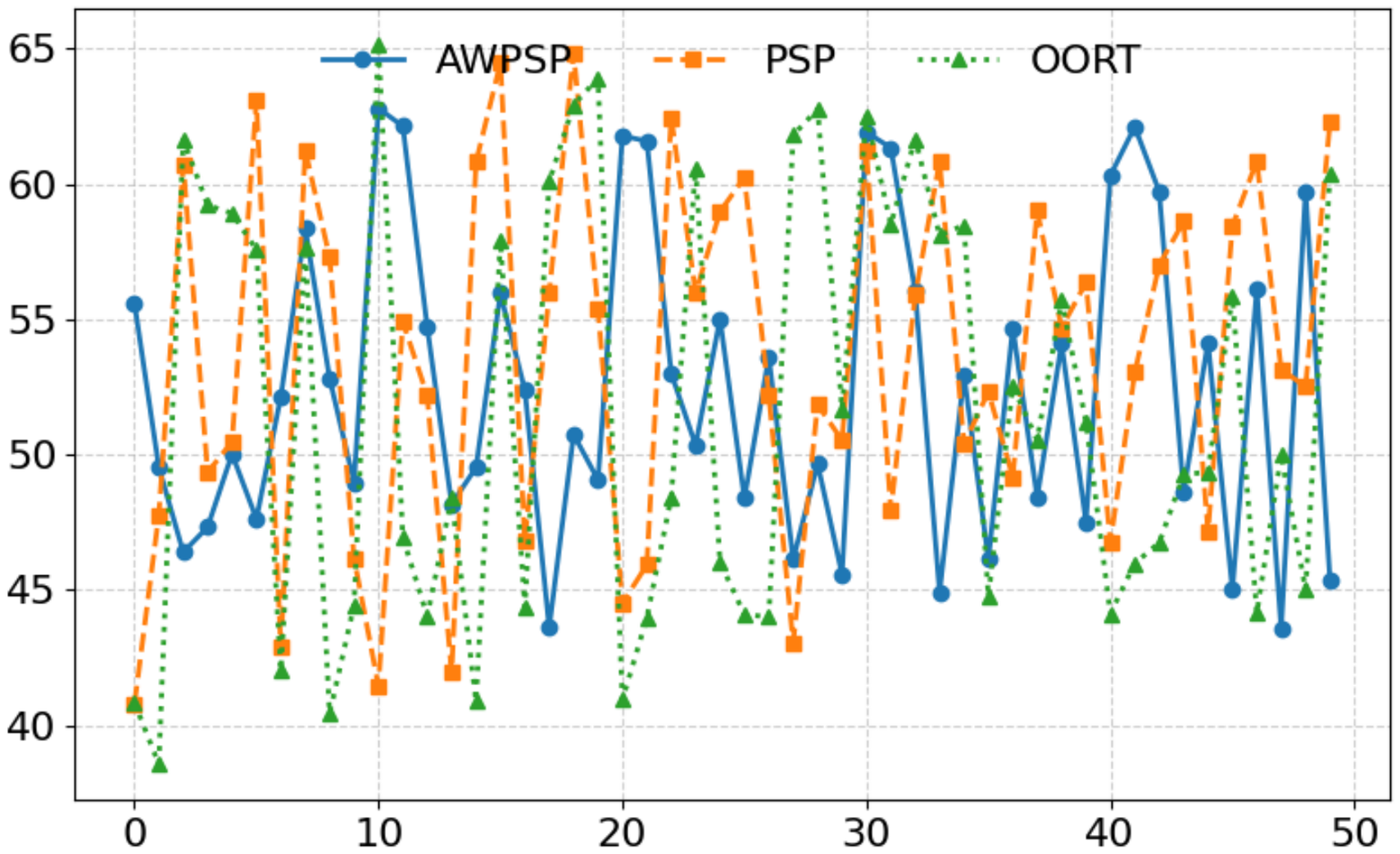}\label{accuracy-5labels}}
    \subfigure[Accuracy test for 10 Labels/Client]{\includegraphics[width=0.27\textwidth]{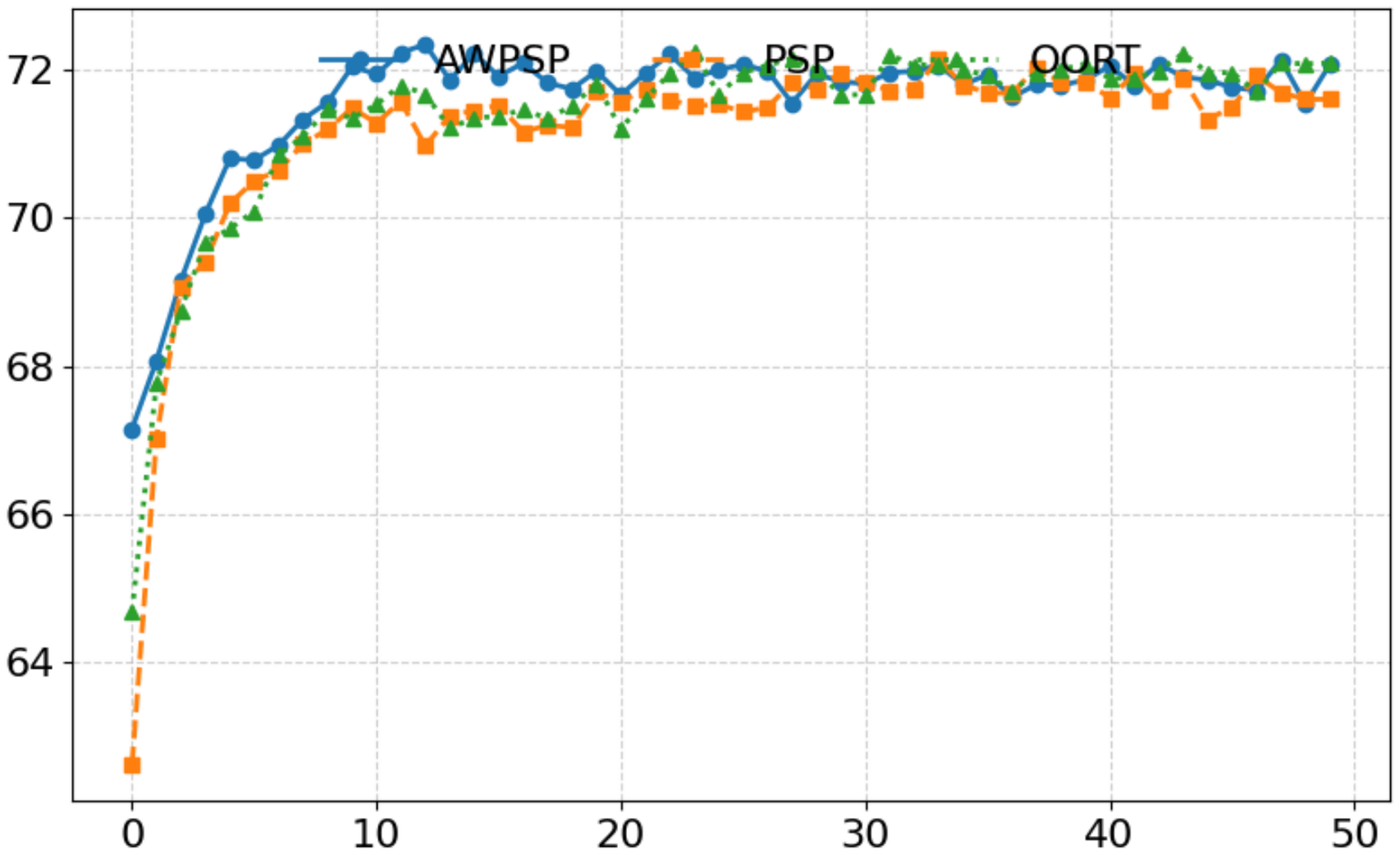}\label{Vaccuracy-5labels}}
      \vspace{-.3cm}
    \caption{Comparison of accuracy for AWPSP/PSP/Oort and different number of labels per client.}
    \label{fig:accuracy-comparison}
\end{figure*}

\subsubsection{Fairness variance metrics}

We define two complementary fairness variance metrics that capture fairness at different levels, see references ~\cite{pei2025fairfl,liu2025fedga}.
At round $t$, let $\mathcal{Y}$ denote the set of all classes, and let
$L_c(t) = \{\ell_{c,1}(t), \dots, \ell_{c,N_c}(t)\}$ be the set of per-sample
losses for class $c$ among the selected nodes $\mathcal{S}_t$.
Define the average loss for class $c$ as:
\begin{equation}
\ell_c(t) = \frac{1}{N_c} \sum_{i=1}^{N_c} \ell_{c,i}(t).
\end{equation}

\paragraph{(a) Average of Within Class Variance or Avg(class-var) or dispersion.}
This measures the heterogeneity of losses \emph{inside each class},
then averages across the observed classes:
\begin{equation}
    \sigma^2_{\text{within-class}}(t)
    = \frac{1}{|\mathcal{Y}_t|}
    \sum_{c \in \mathcal{Y}_t}
    \mathrm{Var}\!\big( L_c(t) \big),
\end{equation}
where $\mathcal{Y}_t \subseteq \mathcal{Y}$ is the set of classes covered
by $\mathcal{S}_t$ in round $t$.
High values indicate that samples of the same class, observed across different
nodes, exhibit highly inconsistent behavior.
AW-PSP might \emph{increase} this metric, because it selects a more diverse
set of devices per round, leading to greater variability.

\paragraph{(b) Variance of Class Means or Var(class-avg) or spread.}
This measures disparity in average performance \emph{across classes}:
\begin{equation}
    \sigma^2_{\text{across-classes}}(t)
    = \mathrm{Var}\Big( \{ \ell_c(t) : c \in \mathcal{Y}_t \} \Big).
\end{equation}
High values indicate that some classes systematically experience higher
losses than others, reflecting a fairness imbalance in training.
AW-PSP reduces this variance, since its availability- and freshness-aware
selection increases label coverage and balances updates across classes,
pulling class-level performance closer together.
Both client selection strategies, AW-PSP and Classic-PSP, are evaluated under identical experimental conditions, including the same model initialization, number of clients selected per round, and failure injection, while varying the degree of data heterogeneity using 2, 5, and 10 labels per client. Their performance is compared against Oort using two complementary class-level fairness metrics: \emph{Avg(class-var)} and \emph{Var(class-avg)}.

Figure~\ref{fig:fairness1-comparison} presents the results for the first fairness metric, \emph{Avg(class-var)}. Across all label configurations, AW-PSP consistently produces the lowest variance growth, indicating more stable within-class behaviour compared to both Classic-PSP and Oort. For 2 labels per client (Fig.~\ref{fairness1-2labels}), all three methods start with similar variance levels, but Oort quickly diverges and exhibits substantially larger increases in later rounds. Classic-PSP shows moderate growth, while AW-PSP maintains the lowest trajectory throughout most of the training process. For 5 labels per client (Fig.~\ref{fairness1-5labels}), Oort demonstrates pronounced instability with repeated spikes, whereas Classic-PSP and AW-PSP follow smoother trends, with AW-PSP remaining consistently lower. When the number of labels increases to 10 per client (Fig.~\ref{fairness1-10labels}), the separation between the three methods becomes clearer: Oort shows the steepest increase in variance, Classic-PSP remains intermediate, and AW-PSP continues to maintain the lowest overall within-class variability.

\begin{figure*}[ht]
    \centering
    \subfigure[Avg(class-var) for 2 Labels/Client]{\includegraphics[width=0.27\textwidth]{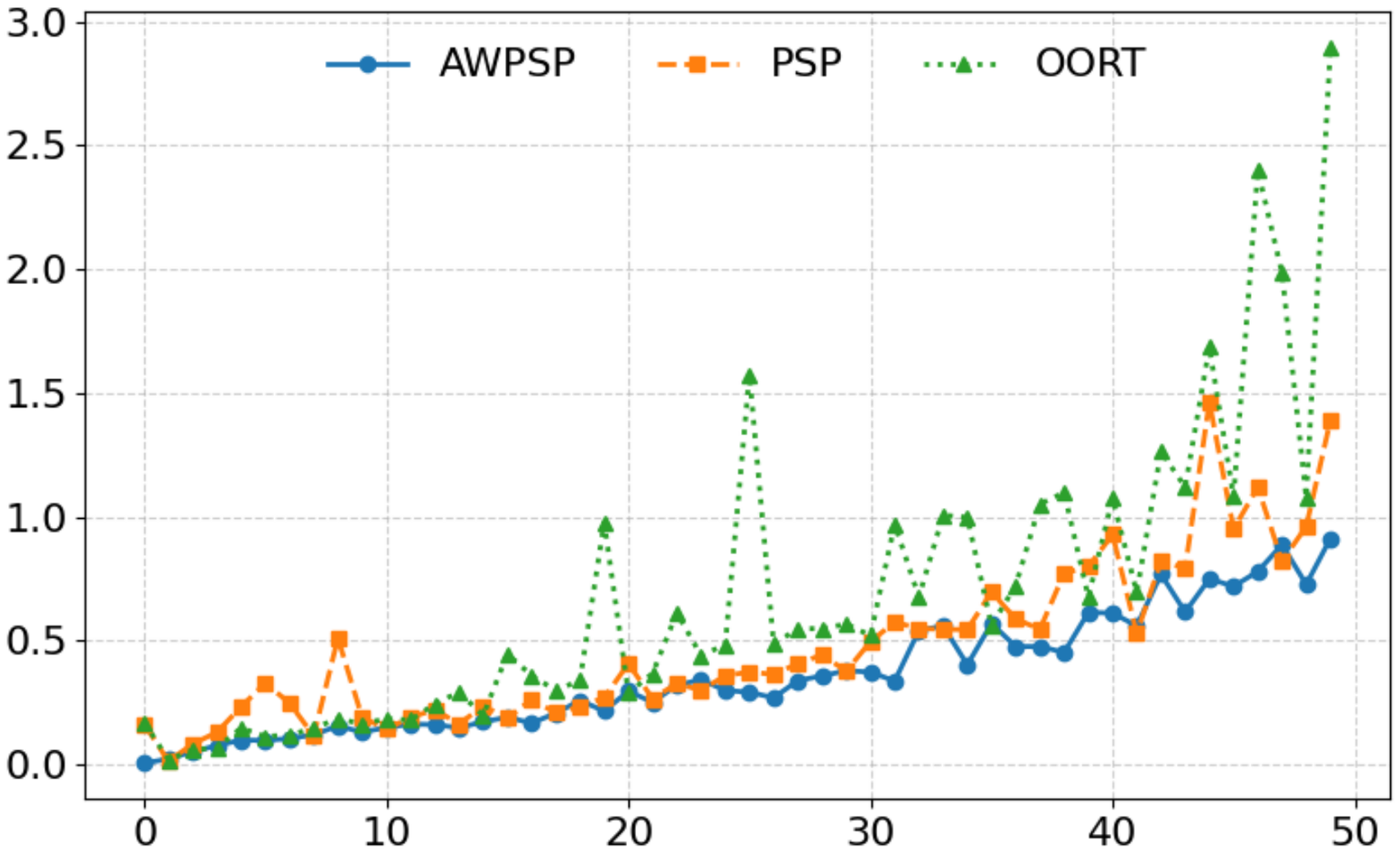}\label{fairness1-2labels}}
    \subfigure[Avg(class-var) for 5 Labels/Client]{\includegraphics[width=0.27\textwidth]{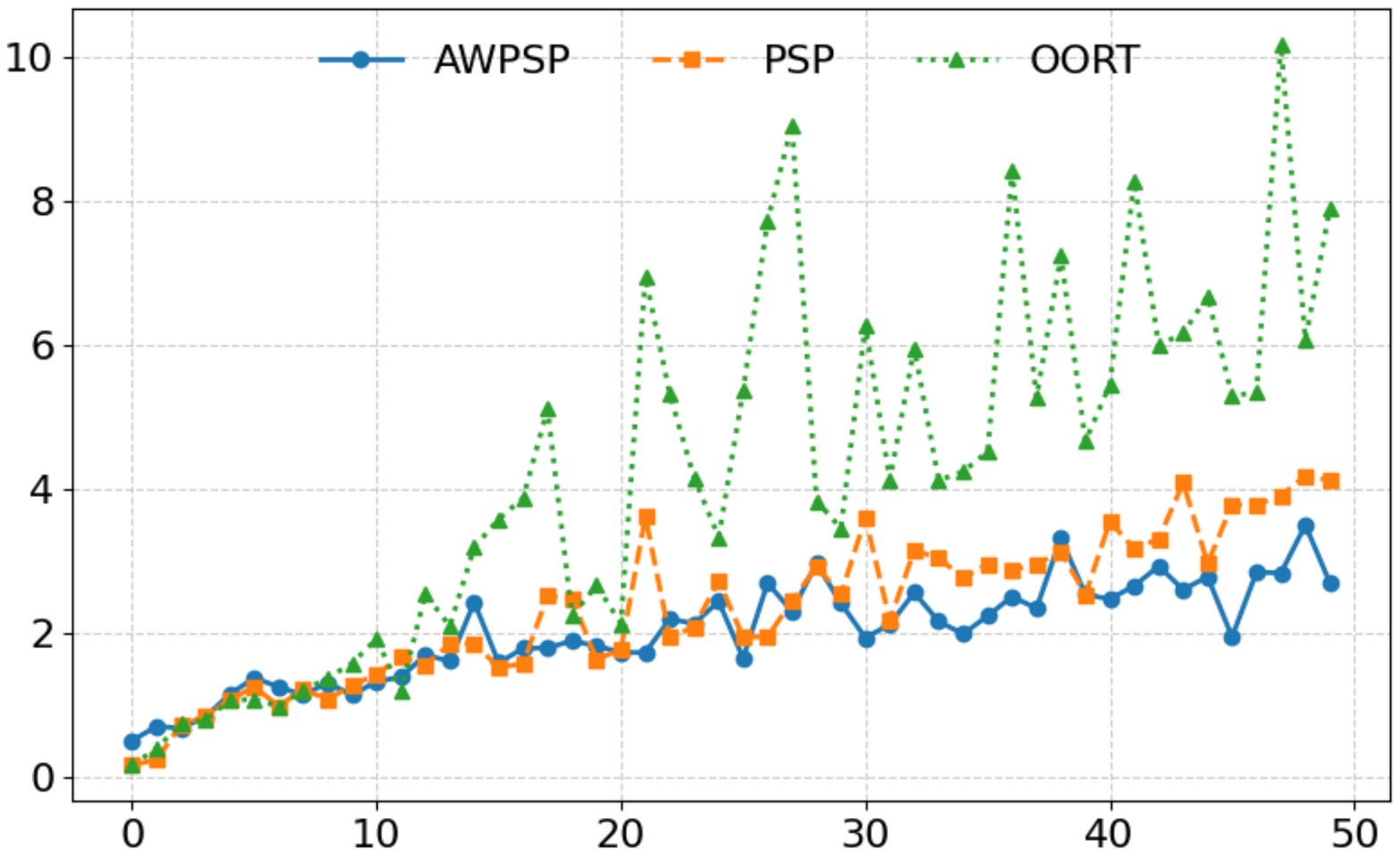}\label{fairness1-5labels}}
    \subfigure[Avg(class-var) for 10 Labels/Client]{\includegraphics[width=0.27\textwidth]{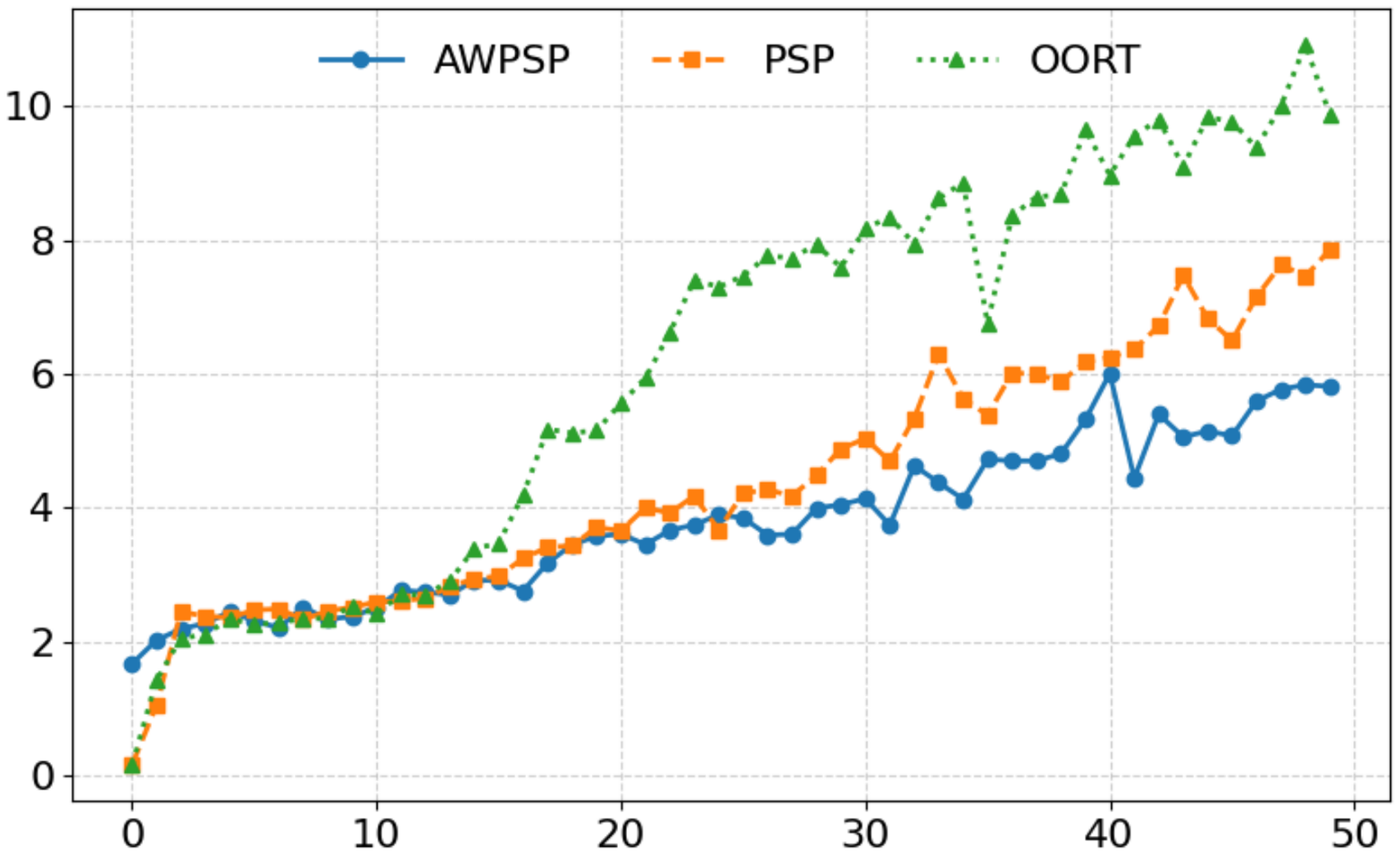}\label{fairness1-10labels}}
    \vspace{-.3cm}
    \caption{Comparison of our 1st fairness metric for AWPSP/PSP/Oort and different number of labels per client.}
    \label{fig:fairness1-comparison}
\end{figure*}

\begin{figure*}[ht]
    \centering
    \subfigure[Var(class-avg) for 2 Labels/Client]{\includegraphics[width=0.27\textwidth]{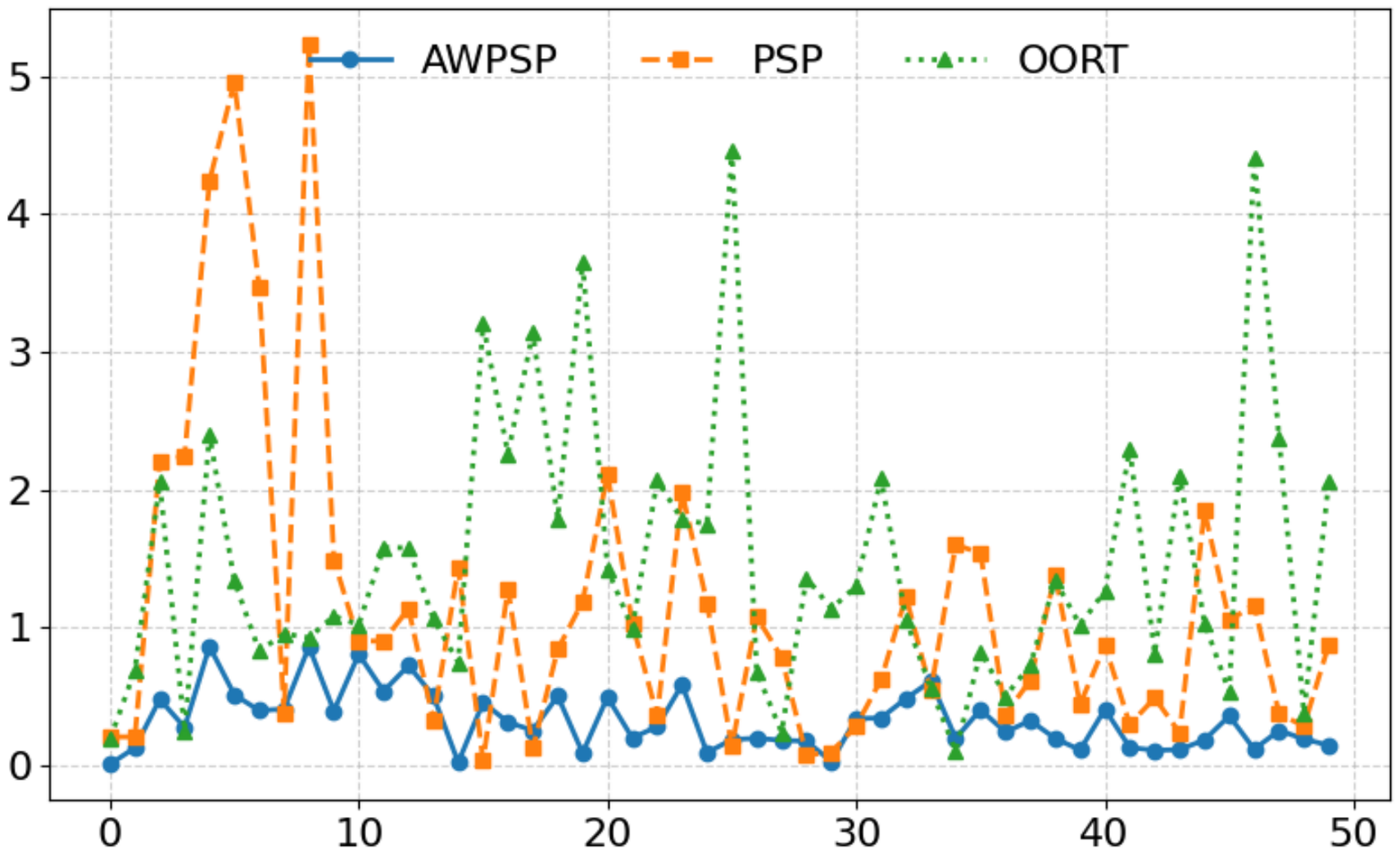}\label{fairness2-2labels}}
    \subfigure[Var(class-avg) for 5 Labels/Client]{\includegraphics[width=0.27\textwidth]{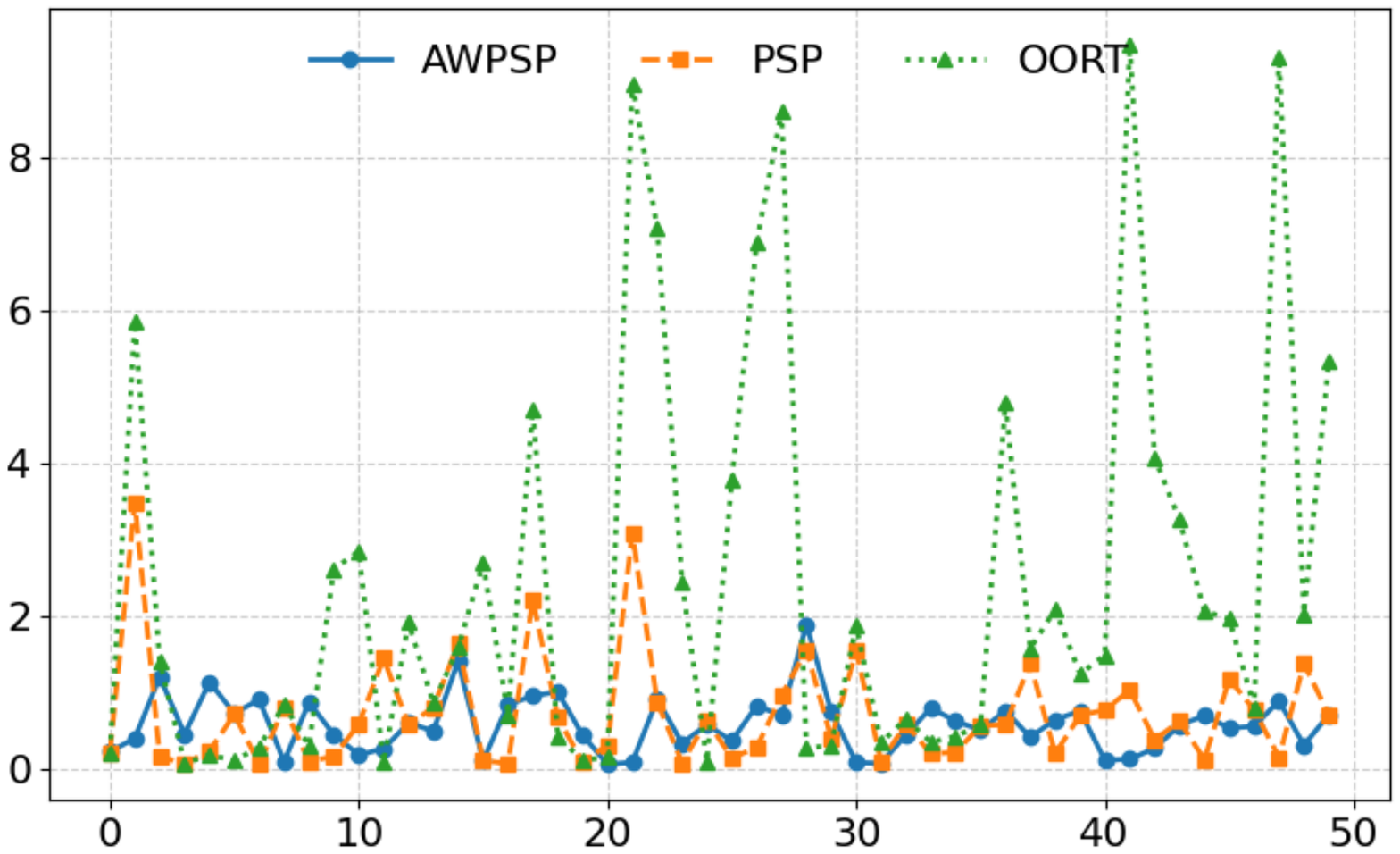}\label{fairness2-5labels}}
    \subfigure[Var(class-avg) for 10 Labels/Client]{\includegraphics[width=0.27\textwidth]{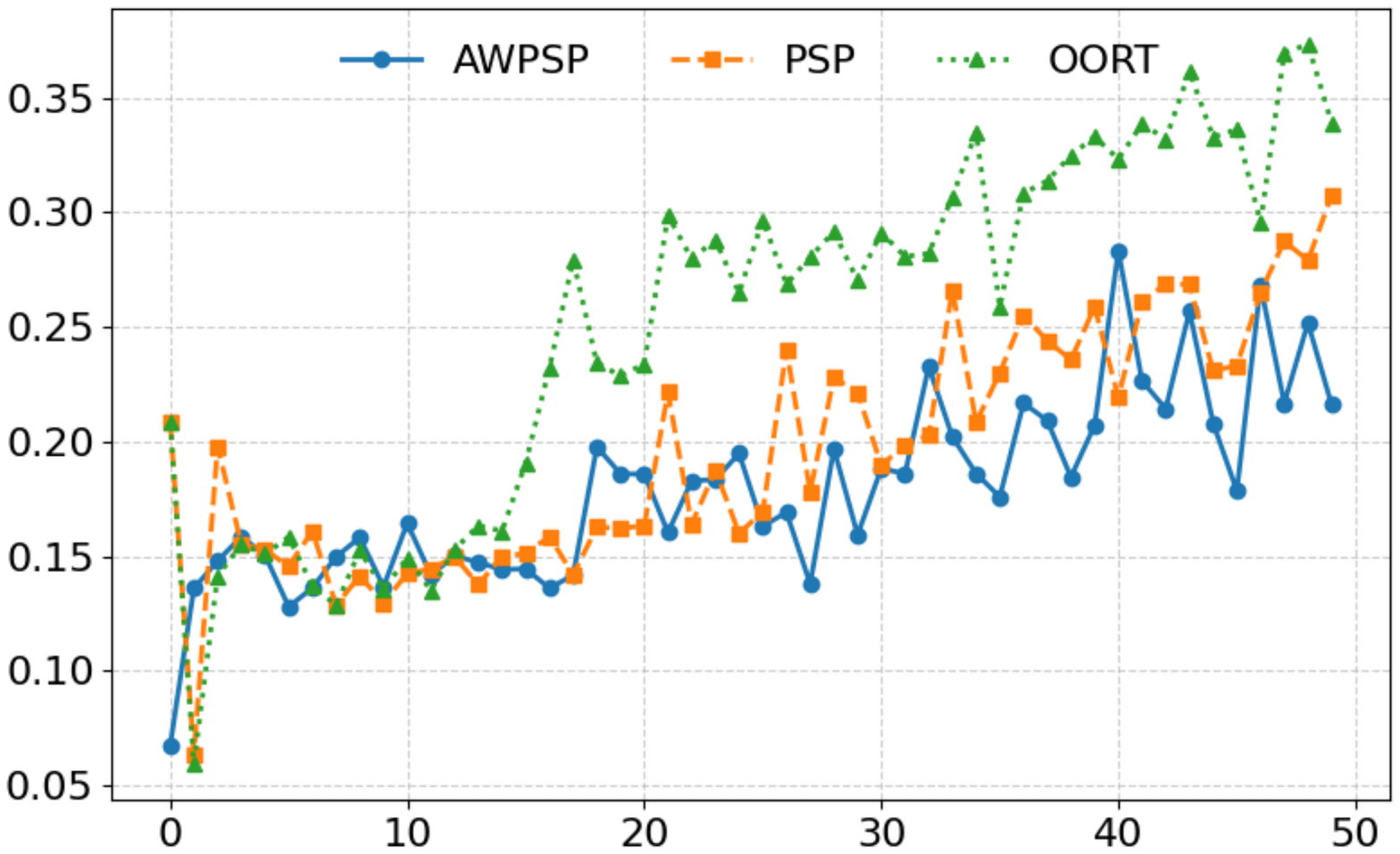}\label{fairness2-10labels}}
        \vspace{-.3cm}
    \caption{Comparison of our 2nd fairness metric for AWPSP/PSP/Oort and different number of labels per client.}
    \label{fig:fairness2-comparison}
\end{figure*}

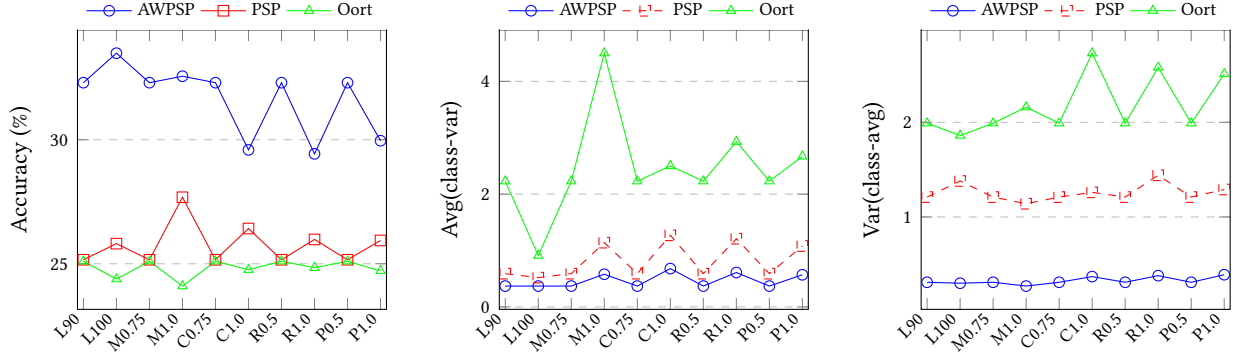
\begin{figure*}[t]
\centering
\begin{tikzpicture}
\begin{groupplot}[
    group style={group size=3 by 1, horizontal sep=1.5cm},
    width=0.32\textwidth,
    height=5.3cm,
    enlarge x limits=0.02,
    xtick=data,
    symbolic x coords={L90,L100,M0.75,M1.0,C0.75,C1.0,R0.5,R1.0,P0.5,P1.0},
    xticklabel style={rotate=45, anchor=east, font=\scriptsize},
    ylabel style={font=\small},
    xlabel style={font=\small},
    tick label style={font=\scriptsize},
    legend style={
        font=\scriptsize,
        at={(0.5,1)},
        anchor=south,
        legend columns=3,
        draw=none
    },
    ymajorgrids=true,
    grid style=dashed
]
\nextgroupplot[
    ylabel={Accuracy (\%)}
]
\addplot+[mark=o] coordinates {
    (L90,32.30) (L100,33.49)
    (M0.75,32.30) (M1.0,32.56)
    (C0.75,32.30) (C1.0,29.59)
    (R0.5,32.30) (R1.0,29.43)
    (P0.5,32.30) (P1.0,29.96)
};
\addplot+[mark=square] coordinates {
    (L90,25.16) (L100,25.81)
    (M0.75,25.16) (M1.0,27.68)
    (C0.75,25.16) (C1.0,26.42)
    (R0.5,25.16) (R1.0,25.98)
    (P0.5,25.16) (P1.0,25.94)
};
\addplot+[mark=triangle, color=green] coordinates {
    (L90,25.10) (L100,24.38)
    (M0.75,25.10) (M1.0,24.10)
    (C0.75,25.10) (C1.0,24.76)
    (R0.5,25.10) (R1.0,24.84)
    (P0.5,25.10) (P1.0,24.71)
};
\legend{AWPSP, PSP, Oort}

\nextgroupplot[
    ylabel={Avg(class-var)}
]
\addplot+[mark=o] coordinates {
    (L90,0.37) (L100,0.37)
    (M0.75,0.37) (M1.0,0.58)
    (C0.75,0.37) (C1.0,0.68)
    (R0.5,0.37) (R1.0,0.61)
    (P0.5,0.37) (P1.0,0.57)
};
\addplot+[mark=square,dashed] coordinates {
    (L90,0.59) (L100,0.52)
    (M0.75,0.59) (M1.0,1.14)
    (C0.75,0.59) (C1.0,1.27)
    (R0.5,0.59) (R1.0,1.21)
    (P0.5,0.59) (P1.0,1.08)
};
\addplot+[mark=triangle, color=green] coordinates {
    (L90,2.23) (L100,0.91)
    (M0.75,2.23) (M1.0,4.5)
    (C0.75,2.23) (C1.0,2.50)
    (R0.5,2.23) (R1.0,2.93)
    (P0.5,2.23) (P1.0,2.67)
};
\legend{AWPSP, PSP, Oort}
\nextgroupplot[
    ylabel={Var(class-avg)}
]
\addplot+[mark=o] coordinates {
    (L90,0.31) (L100,0.30)
    (M0.75,0.31) (M1.0,0.27)
    (C0.75,0.31) (C1.0,0.37)
    (R0.5,0.31) (R1.0,0.38)
    (P0.5,0.31) (P1.0,0.39)
};
\addplot+[mark=square,dashed] coordinates {
    (L90,1.21) (L100,1.38)
    (M0.75,1.21) (M1.0,1.14)
    (C0.75,1.21) (C1.0,1.26)
    (R0.5,1.21) (R1.0,1.44)
    (P0.5,1.21) (P1.0,1.29)
};
\addplot+[mark=triangle, color=green] coordinates {
    (L90,1.99) (L100,1.86)
    (M0.75,1.99) (M1.0,2.16)
    (C0.75,1.99) (C1.0,2.73)
    (R0.5,1.99) (R1.0,2.58)
    (P0.5,1.99) (P1.0,2.51)
};
\legend{AWPSP, PSP, Oort}
\end{groupplot}
\end{tikzpicture}
    \vspace{-.3cm}
\caption{Sensitivity of AWPSP, PSP, and Oort to parameter perturbations across five parameter groups: latency threshold (L), communication weight (M), computation weight (C), recovery probability (R), and correlation penalty (P).}
\label{fig:sensitivity_main_effects}
\end{figure*}

Figure~\ref{fig:fairness2-comparison} shows the second fairness metric, \emph{Var(class-avg)}, which directly reflects imbalance across class-level mean losses. The trends demonstrate a consistent ordering across all label settings. For 2 labels per client (Fig.~\ref{fairness2-2labels}), Classic-PSP exhibits noticeable early-round spikes, indicating unstable class-level fairness, while Oort shows large oscillations across many rounds. AW-PSP remains significantly smoother and lower overall. For 5 labels per client (Fig.~\ref{fairness2-5labels}), Oort again displays extreme peaks, indicating large fluctuations in class-level performance, whereas Classic-PSP shows moderate variability. AW-PSP maintains the lowest and most stable variance trajectory. With 10 labels per client (Fig.~\ref{fairness2-10labels}), all methods show gradual variance growth as training progresses, but AW-PSP consistently stays below Classic-PSP, while Oort exhibits the highest and fastest increase.

Figure~\ref{fig:sensitivity_main_effects} shows that AWPSP consistently achieves the highest accuracy while maintaining significantly lower variance across both fairness metrics compared to PSP and Oort. Notably, performance remains stable under latency, communication, and computation changes, indicating robustness to these parameters. In contrast, recovery and correlation perturbations introduce more pronounced effects, particularly increasing variance for PSP and Oort, while AWPSP degrades more gracefully.

\subsubsection{Other fairness measures}

We evaluate AW-PSP against Classic-PSP and Oort along three complementary
dimensions capturing class coverage, distributional balance, and participation fairness.
At round $t$, let $\mathcal{Y}$ denote the set of all classes.
For each class $c \in \mathcal{Y}$, let $L_c(t) = \{\ell_{c,1}(t), \dots, \ell_{c,N_c}(t)\}$
be the set of per-sample losses associated with class $c$ among the selected clients $\mathcal{S}_t$,
where $N_c$ is the number of samples of class $c$ observed at round $t$.
We define the empirical class distribution $P_t(c)
=
\frac{N_c}{\sum_{c' \in \mathcal{Y}} N_{c'}}$. Let the uniform distribution over classes be $U(c) = \frac{1}{|\mathcal{Y}|}$. See references ~\cite{liu2025fedga,shi2021towards}.

\begin{enumerate}

\item \textbf{KL divergence (distributional imbalance).}

We quantify deviation from uniform class participation via:
\begin{equation}
D_{\mathrm{KL}}(P_t \parallel U)
=
\sum_{c \in \mathcal{Y}}
P_t(c)
\log
\frac{P_t(c)}{U(c)}
\end{equation}

This metric captures how skewed the selected data distribution is.
Lower values indicate better class balance and improved representativeness.

\item \textbf{Unseen classes (coverage gap).}

We measure the number of classes that are completely absent in the selected set. This metric directly quantifies failure of coverage.
Lower values indicate that more classes are represented in each round.

\begin{equation}
U_{\mathrm{miss}}(t)
=
\left|
\left\{
c \in \mathcal{Y}
\;:\;
N_c = 0
\right\}
\right|
\end{equation}

\item \textbf{Gini coefficient (participation inequality).}

Let $n_i(t)$ denote the cumulative number of times client $i$ has been selected up to round $t$.
Define the Gini coefficient below, which measures inequality in participation across clients over time.
Lower values indicate more uniform participation and improved fairness
\begin{equation}
G(t)
=
\frac{
\sum_{i=1}^{N} \sum_{j=1}^{N}
\left| n_i(t) - n_j(t) \right|
}{
2N \sum_{i=1}^{N} n_i(t)
}
\end{equation}

\end{enumerate}

\begin{figure*}[ht]
    \centering
    \subfigure[KL divergence comparison]{\includegraphics[width=0.27\textwidth]{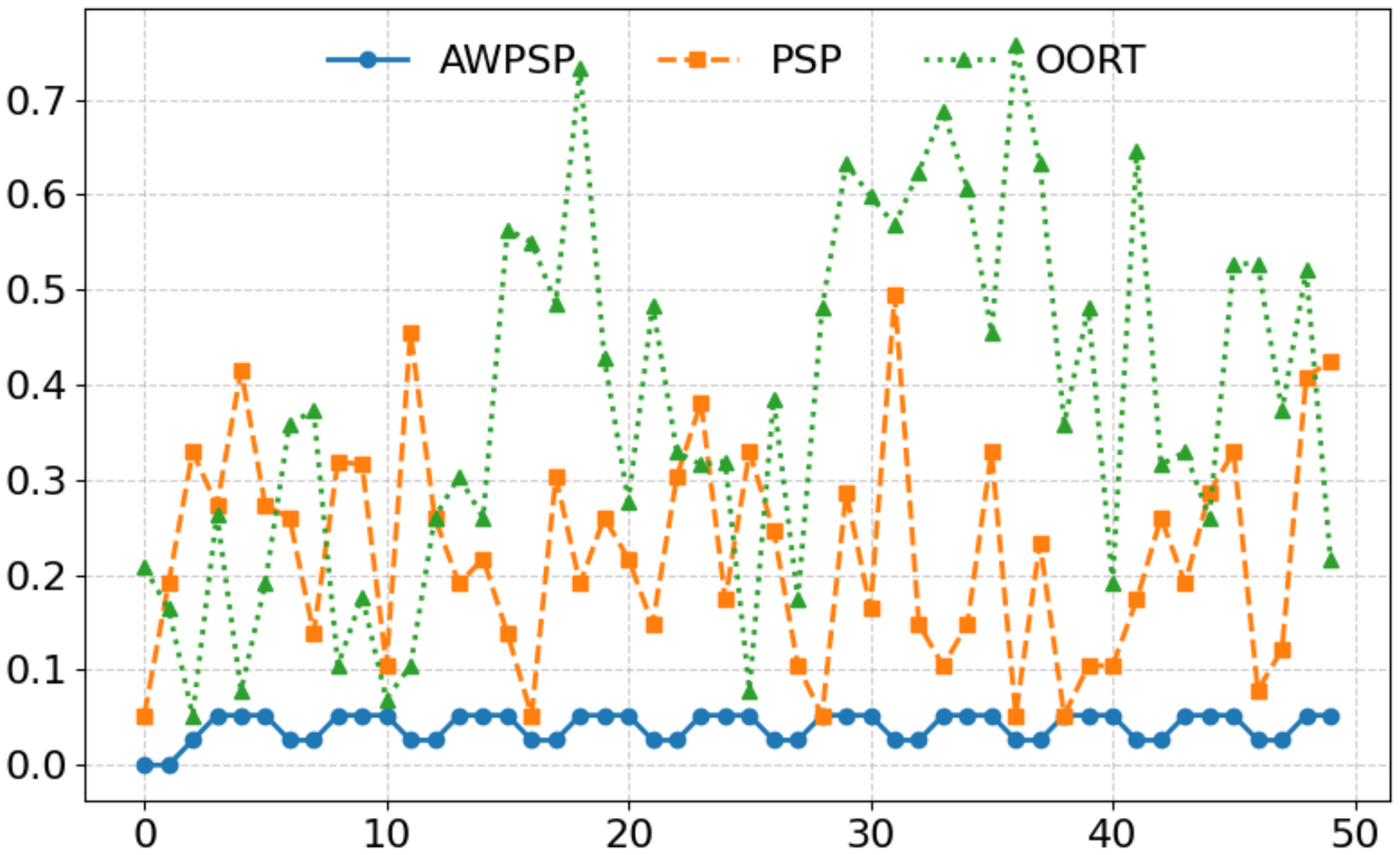}\label{fairness3-KL}}
    \subfigure[Unseen classes comparison]{\includegraphics[width=0.27\textwidth]{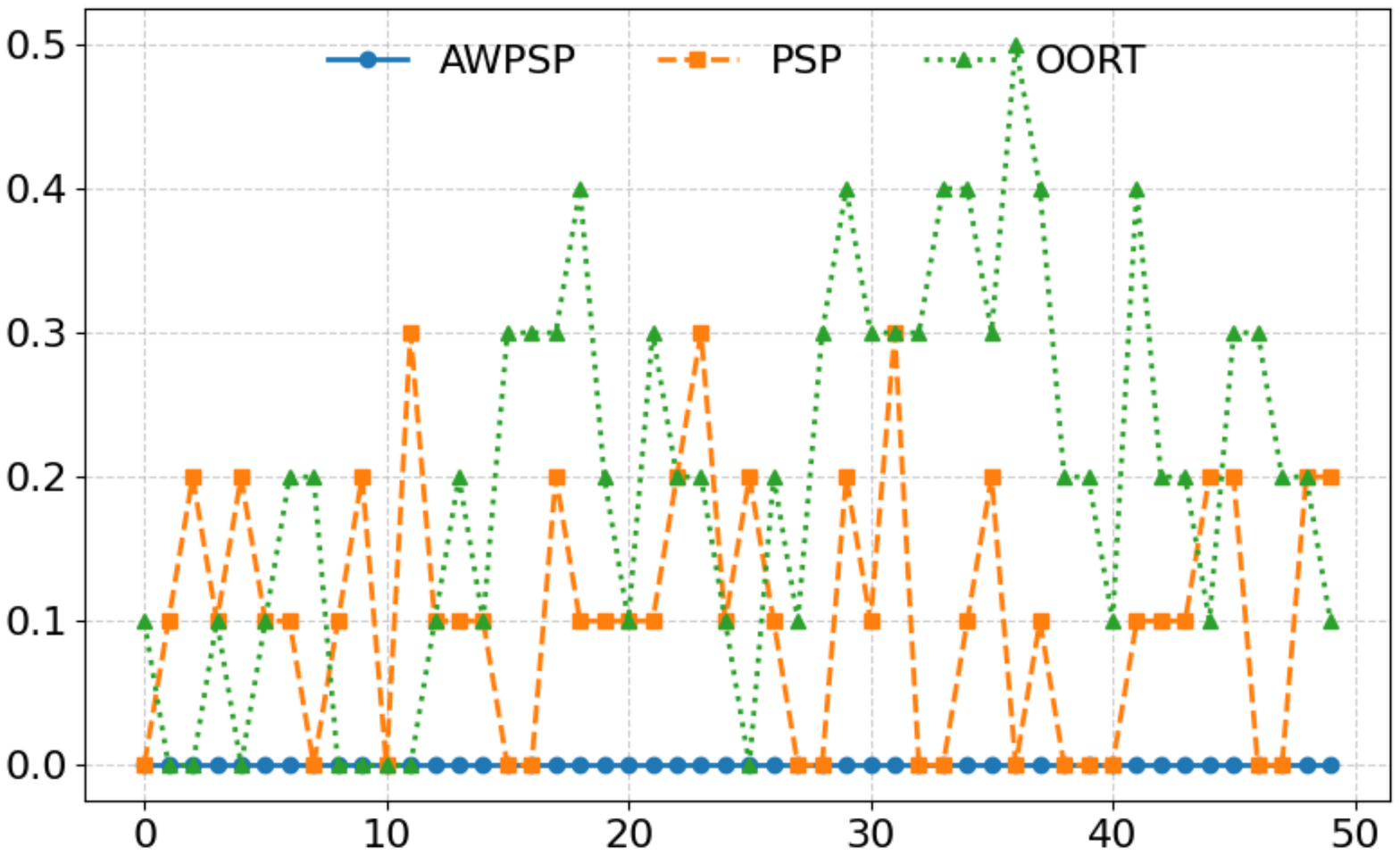}\label{fairness3-Unseen}}
    \subfigure[Gini coeffiicient comparison]{\includegraphics[width=0.27\textwidth]{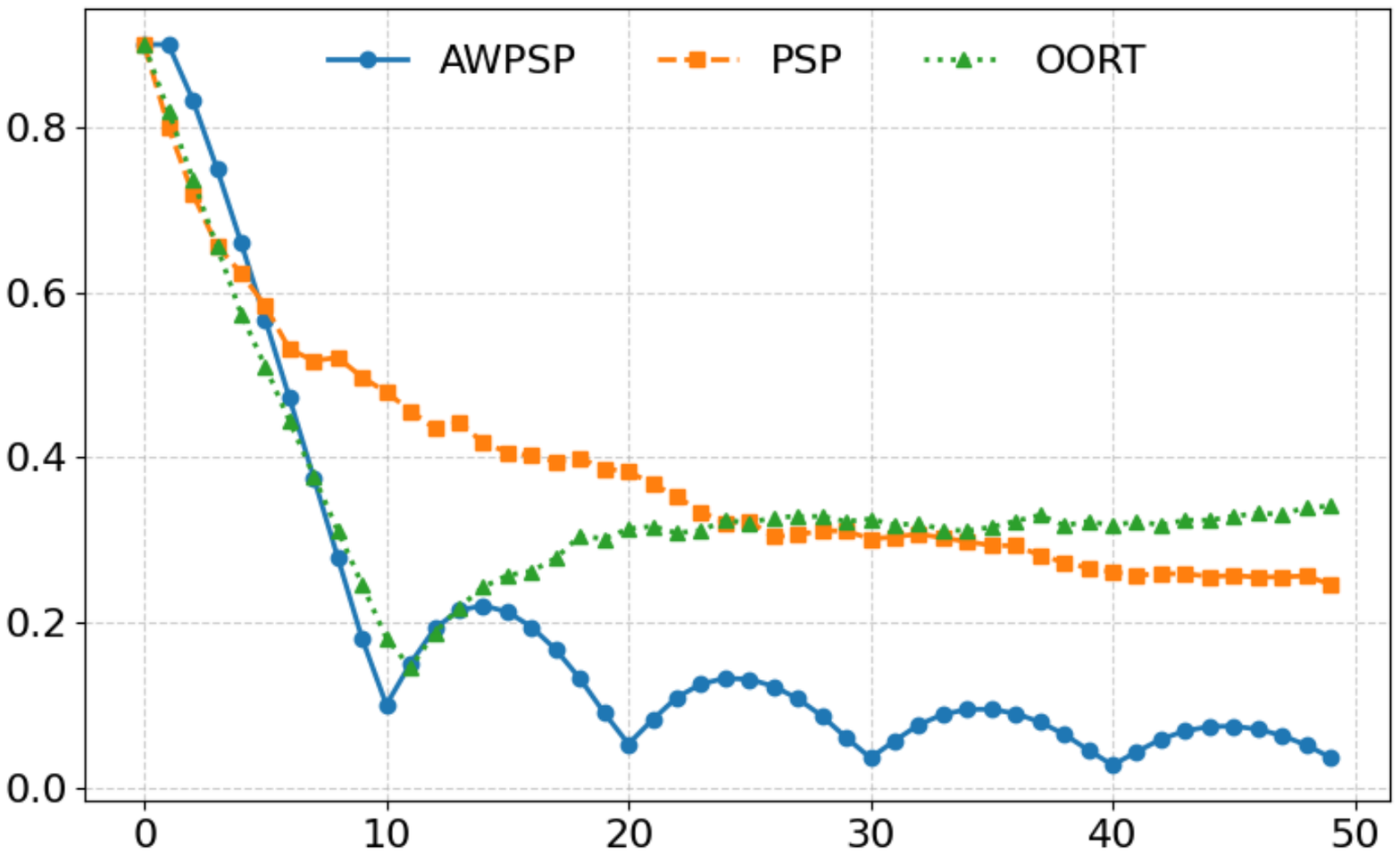}\label{fairness3-Gini}}
      \vspace{-.3cm}
    \caption{Comparison of AWPSP/PSP/Oort in terms of KL divergence, Unseen classes (coverage gap), and Gini coefficient.}
    \label{fig:fairness3-comparison}
\end{figure*}

Figure~\ref{fig:fairness3-comparison} demonstrates the comparison between AW-PSP, Classic-PSP, and Oort across the three fairness-related metrics. Together, these metrics provide strong evidence that AW-PSP consistently improves fairness of client selection, leading to more balanced and representative training data across rounds.
Looking first at KL divergence, AW-PSP maintains consistently low values throughout the training process, typically remaining close to zero, indicating that the class distribution of selected clients closely matches the global distribution. In contrast, Classic-PSP shows larger fluctuations, while Oort frequently produces significantly higher KL values, often exceeding 0.5. This indicates that AW-PSP more effectively mitigates distribution skew, ensuring that no subset of labels dominates the training process.
Next, examining the number of unseen classes, AW-PSP maintains values near zero across nearly all rounds, meaning that almost all classes are represented in each training round. Classic-PSP occasionally fails to cover some classes, while Oort exhibits the largest variation, with unseen-class ratios reaching as high as 0.5 in some rounds. This demonstrates that AW-PSP achieves significantly better label coverage.
Finally, the Gini coefficient, which measures inequality in class representation, further confirms the advantage of AW-PSP. AW-PSP rapidly reduces the Gini value and stabilizes at very low levels (close to 0.05), indicating highly balanced participation of different classes. Classic-PSP converges more slowly and stabilizes at higher inequality levels (around 0.25), while Oort remains consistently higher (around 0.30–0.35), suggesting greater imbalance in selected client data.

\subsubsection{Scalability}

\begin{figure*}[ht]
    \centering
    \subfigure[AW-PSP Accuracy scalability]{\includegraphics[width=0.27\textwidth]{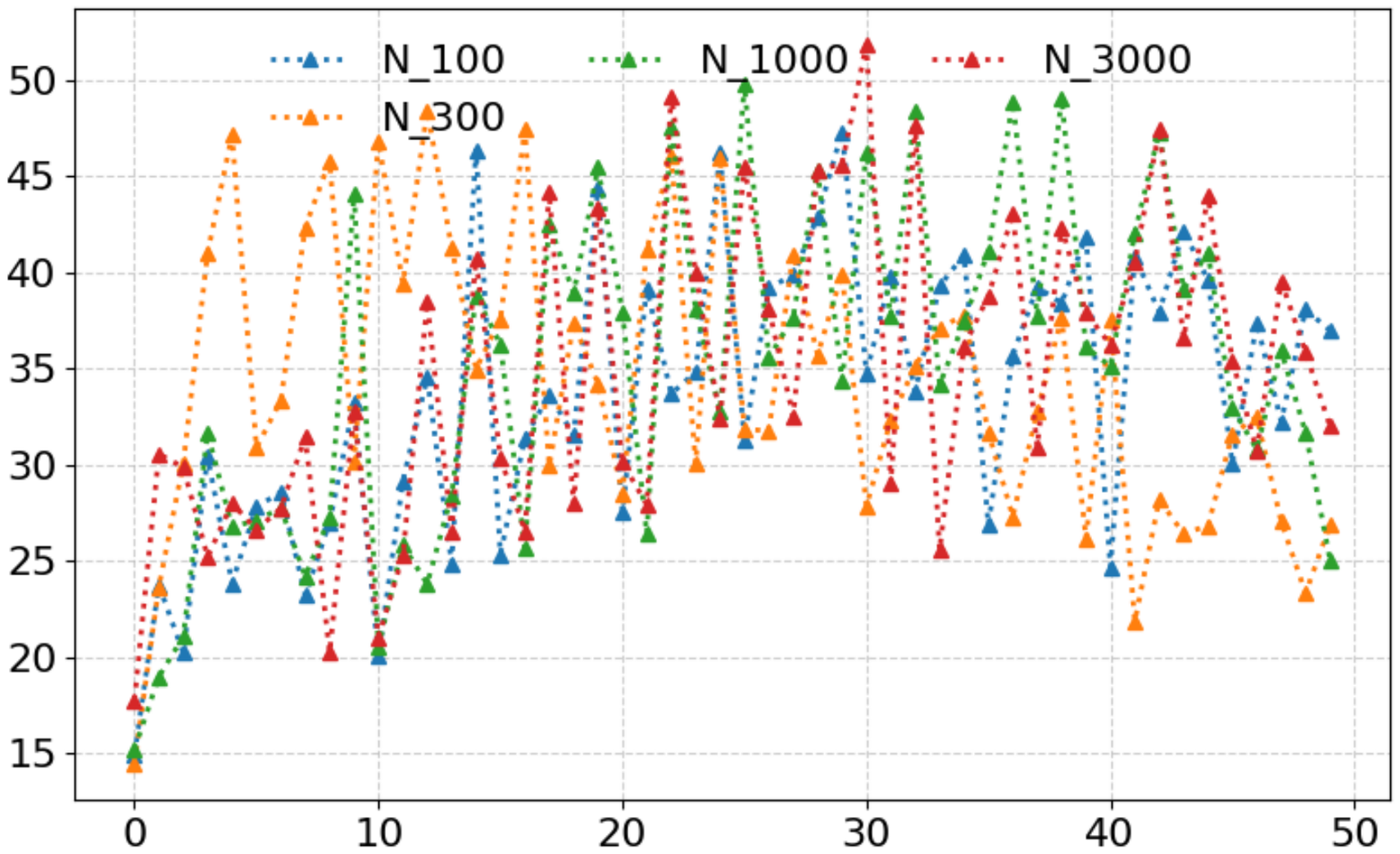}\label{Accuracy_scalability}}
    \subfigure[AW-PSP Avg(class-var) scalability]{\includegraphics[width=0.27\textwidth]{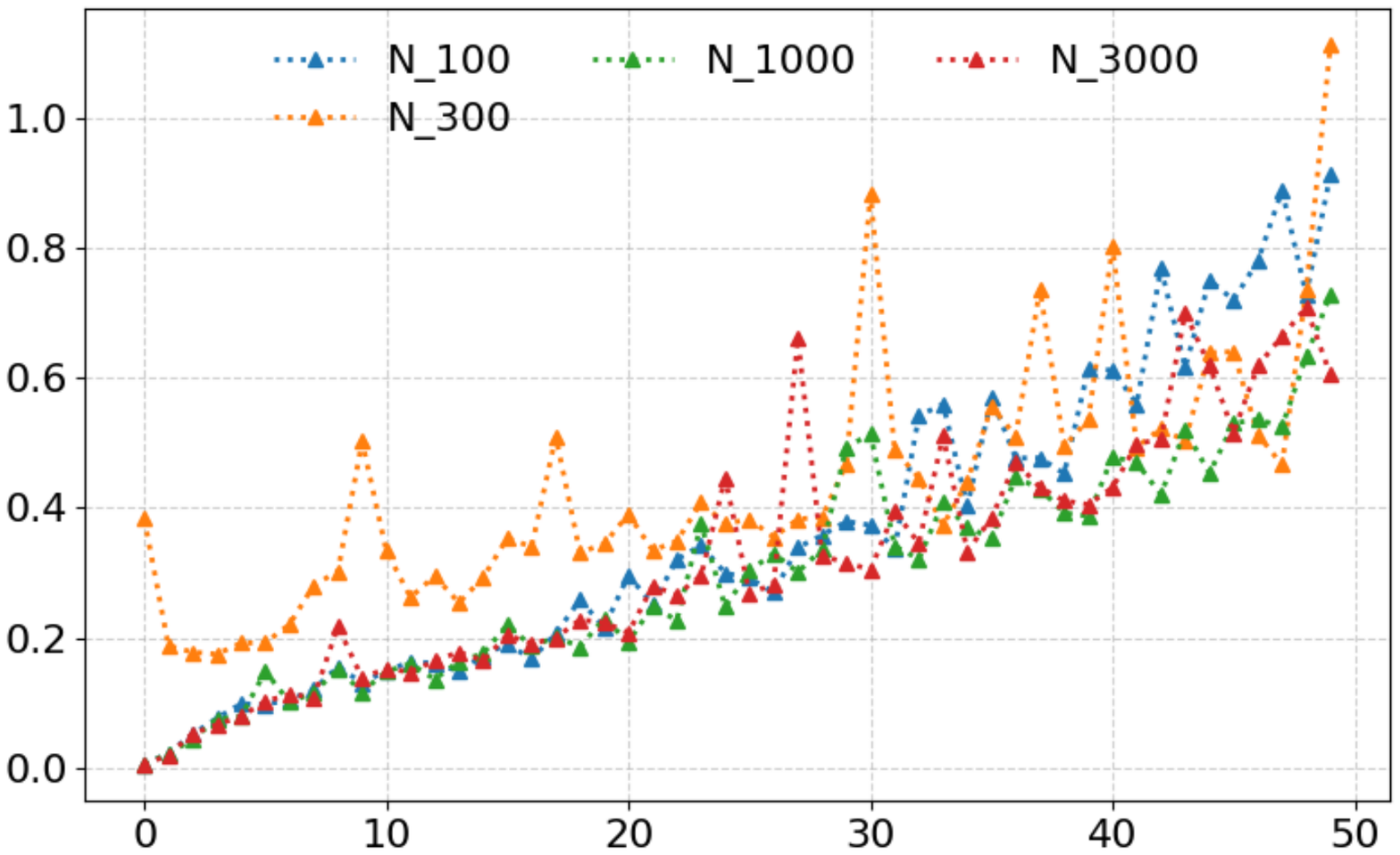}\label{Fairness1_scalability}}
    \subfigure[AW-PSP Var(class-avg) scalability]{\includegraphics[width=0.27\textwidth]{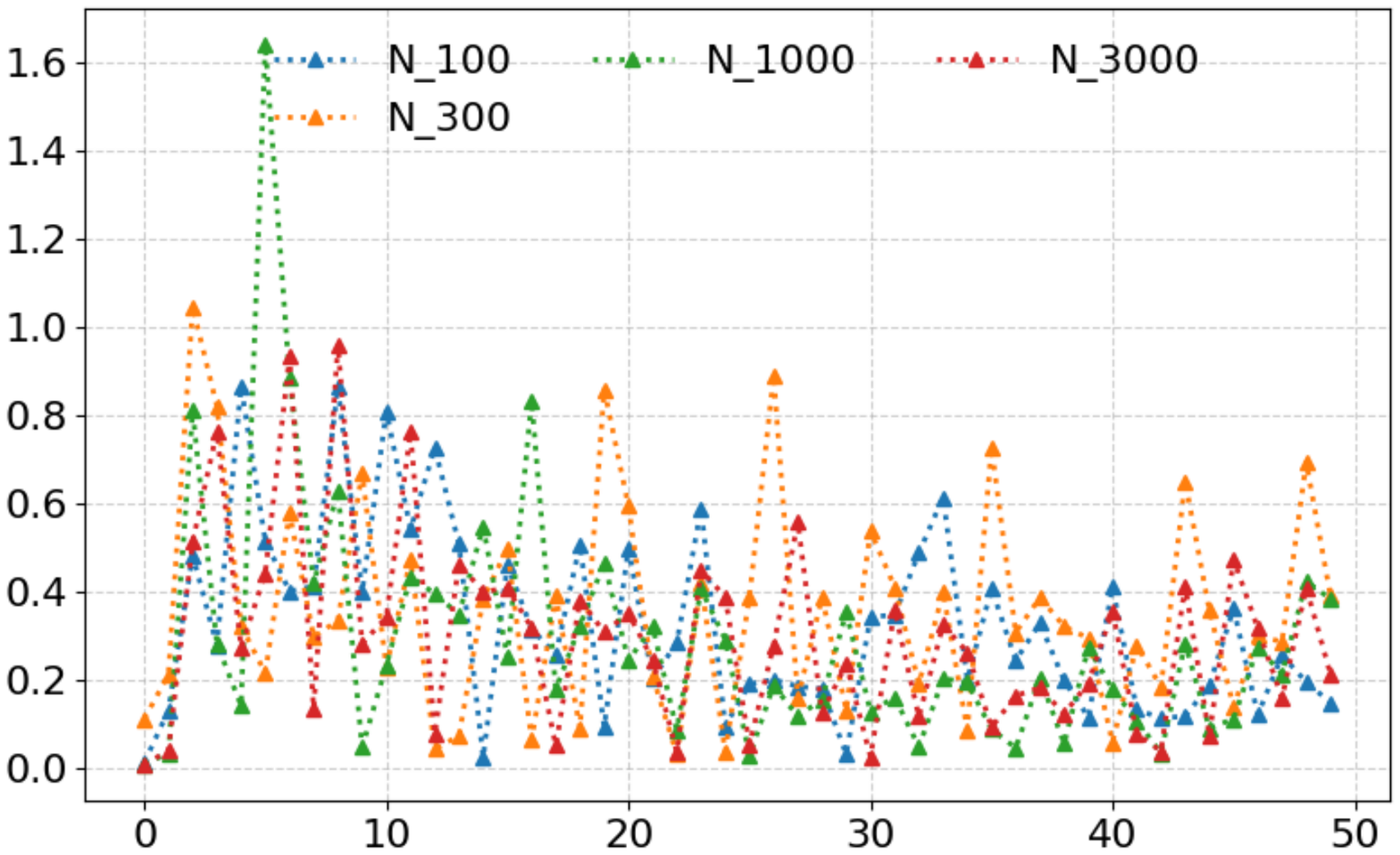}\label{Fairness2_scalability}}
      \vspace{-.3cm}
    \caption{Scalability in terms of accuracy, fairness metric 1 and fairness metric 2 for [100,300,1000,3000] clients.}
    \label{fig:AW-PSP Scalibility}
\end{figure*}

\begin{table*}[t]
\centering
\footnotesize
\caption{Impact of Correlation Noise (2 Labels per Client)}
\label{tab:correlation_noise}
\begin{tabular}{llccccccc}
\hline
Metric & Method & c0 & c10 & c20 & c40 & $\Delta_{c10}$ & $\Delta_{c20}$ & $\Delta_{c40}$ \\
\hline
Accuracy (\%)
& AWPSP & 33.75 & 33.32 & 31.42 & 29.78 & -1.3\% & -6.9\% & -11.8\% \\
& PSP   & 27.85 & 26.77 & 27.53 & 26.84 & -3.9\% & -1.1\% & -3.6\% \\
& Oort  & 24.65 & 24.02 & 23.82 & 24.92 & -2.6\% & -3.4\% & +1.1\% \\
\hline
Avg(within-class)
& AWPSP & 0.36 & 0.38 & 0.58 & 0.58 & +5.6\% & +61.1\% & +61.1\% \\
& PSP   & 0.47 & 0.53 & 0.98 & 1.10 & +12.8\% & +108.5\% & +134.0\% \\
& Oort  & 0.68 & 1.22 & 1.69 & 3.08 & +79.4\% & +148.5\% & +352.9\% \\
\hline
Var(class-avg)
& AWPSP & 0.33 & 0.34 & 0.34 & 0.38 & +3.0\% & +3.0\% & +15.2\% \\
& PSP   & 1.16 & 1.22 & 1.35 & 1.29 & +5.2\% & +16.4\% & +11.2\% \\
& Oort  & 1.47 & 1.85 & 2.02 & 3.09 & +25.9\% & +37.4\% & +110.2\% \\

\hline
\end{tabular}
\end{table*}

To evaluate the scalability of AW-PSP under highly heterogeneous conditions, we analyze its behavior when the number of clients increases from 100 to 300, 1000, and 3000, while each client contains only 2 labels. This represents an extremely non-IID scenario, where each individual client provides very limited class diversity, making fair and representative client selection particularly challenging. In such conditions, we expect a robust selection strategy to maintain stable accuracy while preserving balanced class representation as the client pool grows.
Looking first at model accuracy, Figure~\ref{Accuracy_scalability} shows that AW-PSP maintains relatively stable performance as the number of clients increases. While the absolute accuracy remains lower than in less heterogeneous settings due to the extreme label skew, the curves for larger populations (1000 and 3000 clients) remain comparable to those observed with 100 and 300 clients. This indicates that AW-PSP successfully scales without significant degradation in predictive performance. The stability of the accuracy curves suggests that increasing the client pool does not introduce additional bias, demonstrating that the selection mechanism continues to identify useful and representative participants even when label diversity per client is highly constrained.
Next, we examine Average of Class Variance, which measures variability within each class across participating clients. As shown in Figure~\ref{Fairness1_scalability}, the average within-class variance increases gradually as the number of clients grows. This trend is expected: with more clients available, the sampling process captures a wider range of feature distributions for each class. Importantly, the growth is smooth and controlled rather than abrupt, indicating that AW-PSP maintains consistent representation of each class despite the limited label availability per client. This behavior suggests that AW-PSP effectively balances exploration of diverse clients while preserving stability in class-specific learning.
Finally, we analyze the Variance of Class Means; Figure~\ref{Fairness2_scalability} shows that it remains relatively low and stable across different client scales, with only moderate fluctuations as the number of clients increases. This indicates that AW-PSP continues to distribute participation opportunities evenly across classes even when the candidate pool is large.

\subsubsection{Ablation (Correlation Impact)}

Correlation noise (i.e., failure correlation among clients) introduces dependencies among selected participants and affects all methods; however, its impact differs significantly across selection strategies. As shown in Table~\ref{tab:correlation_noise}, AW-PSP exhibits a gradual and consistent degradation in accuracy as correlation increases, with only $-1.3\%$ at $c10$, $-6.9\%$ at $c20$, and $-11.8\%$ at $c40$. In contrast, PSP shows smaller but irregular fluctuations, while Oort remains unstable and does not follow a consistent degradation trend, even slightly improving at high correlation.
The differences become substantially more pronounced when examining fairness metrics. For \textit{Avg(within-class)} (intra-class variability), AW-PSP shows a controlled increase, rising moderately from $0.36$ to $0.58$ ($+61.1\%$), indicating bounded degradation in within-class consistency. In contrast, PSP exhibits significantly larger growth ($+134.0\%$), while Oort demonstrates extreme sensitivity, with intra-class variability increasing sharply from $0.68$ to $3.08$ ($+352.9\%$).
A similar pattern is observed for \textit{Var(class-avg)} (inter-class fairness). AW-PSP remains relatively stable, with only a modest increase ($+3.0\%$ to $+15.2\%$), indicating robustness in maintaining balanced class-level performance under correlated failures. PSP shows moderate degradation, while Oort again exhibits substantial imbalance, with variance increasing by up to $+110.2\%$ at $c40$.

\section{Related Work}
\label{s:related}

Prior work in distributed and federated learning has largely focused on stragglers, non-IID data, and heterogeneity, but has not treated correlated availabilities and co-failures as first-class challenges. Classic scheduling approaches such as PSP ~\cite{wang2017probabilistic}, SSP ~\cite{ho2013more}, and gradient coding ~\cite{tandon2017gradient} mitigate slow or failing workers, but assume independent failures and do not address systematic, infrastructure-driven unavailability. In the FL domain, methods like FedAvg ~\cite{mcmahan2017communication}, FedProx~\cite{li2020federated}, FedGS ~\cite{FedGS}, and Oort ~\cite{Oort} improve efficiency or selection under heterogeneous data, but they optimize accuracy–latency trade-offs without considering fairness imbalance or correlated outages. Recent work such as Sageflow ~\cite{sageflow2021} and FLuID ~\cite{fluid2023} handle stragglers and adversaries, while FedScale ~\cite{FedScale} provides large-scale benchmarking of heterogeneity and resource constraints, yet none of these lines systematically model failure correlations or fairness starvation as we do. Even availability-oriented studies such as CA-Fed ~\cite{CA-Fed} or F3AST ~\cite{F3AST} restrict attention to independent client churn, ignoring the clustered or synchronized unavailability patterns that naturally arise in real systems. AW-PSP complements existing frameworks but introduces a distinct failure- and fairness-aware probabilistic sampling that remains effective even when availability shocks are correlated across groups.

\section{Conclusion}
\label{s:conclusions}

This work introduces AW-PSP, a client selection framework that integrates availability prediction, correlated-failure modeling, and DHT-based proximity into FL scheduling. Unlike prior approaches, AW-PSP explicitly models data-availability co-correlation, capturing both behavioral similarity in availability traces and structural co-failure risks arising from shared infrastructure. By incorporating correlation-aware recovery probabilities into a health score defined as \textit{availability(t) $\times$ freshness (t)}, AW-PSP proactively diversifies client selection while still prioritizing reliable nodes, thereby reducing the likelihood of synchronized dropouts.
Extensive trace-driven evaluation demonstrates that AW-PSP consistently outperforms Classic-PSP and Oort across multiple dimensions. First, AW-PSP improves fairness by maintaining lower Var(class-avg) imbalance, near-zero unseen classes, and reduced KL divergence, indicating more balanced and representative participation across labels. Second, AW-PSP preserves class coverage and stabilizes Avg(class-var), leading to more consistent convergence even under highly non-IID settings (2 labels per client). Third, under correlated failures, AW-PSP shows significantly higher robustness where impact of correlation noise on accuracy and fairness is substantially smaller compared to PSP with larger degradation.
Furthermore, scalability experiments with up to 3000 clients show that AW-PSP maintains stable accuracy and fairness.

\bibliographystyle{IEEEtran}
\bibliography{references}

\newpage
\appendix

\section{Theoretical Justification and Guarantees for the Model in Section 3}

This appendix rigorously validates the availability construction
introduced in Section~3, specifically Eq.~(8)--(10).
We derive these expressions from an explicit generative model of
client participation and show that the proposed formulation
is statistically optimal or a controlled approximation under
clearly stated assumptions.

\subsection{Generative Model of Round Success}

For client $i$ at round $t$, define the binary success variable:

\begin{equation}
Y_i(t) \in \{0,1\}
\end{equation}

where $Y_i(t)=1$ indicates successful completion of the FL round.
Let the client resource vector be:

\begin{equation}
R_i(t) =
\begin{pmatrix}
C_i(t) \\
B_i(t)
\end{pmatrix}
\end{equation}
where $C_i(t)$ and $B_i(t)$ denote available computation and communication resources.
Let the minimum required resources for participation be
\begin{equation}
R_{\min} =
\begin{pmatrix}
C_{\min} \\
B_{\min}
\end{pmatrix}
\end{equation}

A node is resource-feasible if
\begin{equation}
Y_i(t) = \mathbf{1}\{C_i(t) \ge C_{\min},\; B_i(t) \ge B_{\min}\}
\end{equation}

Computation and communication delays are functions of the available resources,
i.e., $T_i^{\text{comp}} = f(C_i(t))$ and $T_i^{\text{comm}} = g(B_i(t))$.
Node availability is then defined as the probability of completing within a deadline:
\begin{equation}
a_i(t) = P\big(T_i^{\text{comp}}(t) + T_i^{\text{comm}}(t) \leq T_{\max}\big)
\end{equation}

Participation is feasible iff

\begin{equation}
Y_i(t) = 1
\iff
C_i(t) \ge C_{\min}
\text{ and }
B_i(t) \ge B_{\min}
\end{equation}

Define the true participation probability:

\begin{equation}
p_i^{\text{true}}(t)
=
\mathbb{P}(Y_i(t)=1)
\end{equation}

\subsection{Optimality of Threshold-Based Availability}

In Eq.~(8), computation success at time $t$ is defined as:
\begin{equation}
H_i^{\text{comp}}(t)
=
\mathbf{1}\{C_i(t) \ge C_{\min}\}
\end{equation}

We estimate computation availability using recent history across the last $T$ rounds:
\begin{equation}
a_i^{\text{comp}} = \frac{1}{T} \sum_{t=1}^{T} H_i^{\text{comp}}(t)
\end{equation}

\begin{theorem}[Bayes Optimality of Threshold Rule]
Assume successful computation is equivalent to
$C_i(t) \ge C_{\min}$.
Under 0–1 loss, the estimator
\begin{equation}
\hat{Y}_i^{\text{comp}}(t)
=
\mathbf{1}\{C_i(t) \ge C_{\min}\}
\end{equation}
is the Bayes-optimal classifier for predicting computational feasibility.
\end{theorem}

\begin{proof}
Consider predicting binary variable
\begin{equation}
Z_i(t) = \mathbf{1}\{C_i(t) \ge C_{\min}\}
\end{equation}

Under 0–1 loss,
the Bayes classifier is

\begin{equation}
\hat{Z}_i(t)
=
\mathbf{1}
\left\{
\mathbb{P}(Z_i(t)=1 \mid C_i(t)) \ge \frac{1}{2}
\right\}
\end{equation}

Since $Z_i(t)$ is deterministically defined by $C_i(t)$,

\begin{equation}
\mathbb{P}(Z_i(t)=1 \mid C_i(t))
=
\mathbf{1}\{C_i(t)\ge C_{\min}\}
\end{equation}

Thus the Bayes decision boundary is exactly
$C_i(t)=C_{\min}$.
No alternative rule achieves smaller expected misclassification risk.
Therefore thresholding is not heuristic,
but the unique optimal decision rule under feasibility semantics.
\end{proof}

\subsection{Multiplicative Composition}

Eq.~(9) defines total availability as

\begin{equation}
a_i(t)
=
a_i^{\text{comp}}(t)
\cdot
a_i^{\text{comm}}(t)
\end{equation}

\begin{theorem}[Exact Factorization Under Conditional Independence]
Assume that conditional on system state $S_i(t)$,
computation and communication failures are independent:

\begin{equation}
C_i(t) \perp B_i(t) \mid S_i(t)
\end{equation}

Then

\begin{equation}
p_i^{\text{true}}(t)
=
a_i^{\text{comp}}(t)
\cdot
a_i^{\text{comm}}(t)
\end{equation}
\end{theorem}

\begin{proof}
By definition,

\begin{equation}
p_i^{\text{true}}(t)
=
\mathbb{P}(C_i \ge C_{\min}, B_i \ge B_{\min})
\end{equation}

Condition on $S_i(t)$:

\begin{equation}
=
\mathbb{E}_{S_i}
\left[
\mathbb{P}(C_i \ge C_{\min}, B_i \ge B_{\min} \mid S_i)
\right]
\end{equation}

Under conditional independence,

\begin{equation}
=
\mathbb{E}_{S_i}
\left[
\mathbb{P}(C_i \ge C_{\min} \mid S_i)
\cdot
\mathbb{P}(B_i \ge B_{\min} \mid S_i)
\right]
\end{equation}

If $S_i(t)$ is observed via telemetry,
the conditional expectation collapses and

\begin{equation}
p_i^{\text{true}}(t)
=
a_i^{\text{comp}}(t)
a_i^{\text{comm}}(t)
\end{equation}

Thus multiplication follows directly from probability axioms.
\end{proof}

\subsection{Dependence Case and Approximation Error}

When independence does not hold,
define deviation

\begin{equation}
\delta_i(t)
=
p_i^{\text{true}}(t)
-
a_i^{\text{comp}}(t)
a_i^{\text{comm}}(t)
\end{equation}

\begin{theorem}[Sharp Fréchet Error Bound]
For arbitrary dependence,

\begin{equation}
|\delta_i(t)|
\le
\min(a_c,a_b)
-
a_c a_b
\end{equation}

where $a_c = a_i^{\text{comp}}(t)$ and
$a_b = a_i^{\text{comm}}(t)$
\end{theorem}

\begin{proof}
By Fréchet–Hoeffding bounds:

\begin{equation}
\max(0,a_c+a_b-1)
\le
p_i^{\text{true}}
\le
\min(a_c,a_b)
\end{equation}

Subtract $a_c a_b$ and take absolute value.
The maximum deviation occurs at the endpoints,
giving the stated bound.
\end{proof}

Thus multiplicative composition is the maximum-entropy
joint model with bounded worst-case error.

\subsection{Recovery Probability}

Eq.~(10) defines the recovery probability as the conditional probability that a node
meets the deadline in the next round, given that it failed in the previous round:

\begin{equation}
\begin{aligned}
\beta_i(t)
= &
\Pr\big(
T_i^{\mathrm{comp}}(t+1) + T_i^{\mathrm{comm}}(t+1) \le \\
&
T_{\max}
\;\big|\;
T_i^{\mathrm{comp}}(t) + T_i^{\mathrm{comm}}(t) > T_{\max}
\big)
\end{aligned}
\end{equation}

\begin{theorem}[Recovery under Hazard-Based Model]
Assume that the completion time
$T_i(t) = T_i^{\mathrm{comp}}(t) + T_i^{\mathrm{comm}}(t)$
follows a stochastic process with hazard rate $\lambda_i(t)$ governing
the recovery from deadline violations.
Then, for a small time interval $\Delta$ corresponding to one round,
the recovery probability satisfies
\begin{equation}
\beta_i(t)
=
1 - \exp(-\lambda_i(t)\Delta)
+
O(\Delta^2)
\end{equation}
\end{theorem}

\begin{proof}
Let $T_i$ denote the time until the node returns to a state where it meets the deadline.
The survival function of remaining in the failed state is
\begin{equation}
S_i(t)
=
\exp\left(
-
\int_0^t \lambda_i(s)\, ds
\right)
\end{equation}

Conditioned on the node being in a failed state at time $t$, the probability of recovery
within the next interval $\Delta$ is
\begin{equation}
\Pr(T_i \le \Delta \mid T_i > 0)
=
1 - \exp(-\lambda_i(t)\Delta)
+
O(\Delta^2)
\end{equation}

Interpreting $\Delta$ as the duration of one FL round yields the desired result.
For constant hazard (exponential recovery), the expression is exact.
\end{proof}

\subsection{Global Consistency Theorem}

The AW-PSP model defines:

\begin{equation}
p_i'(t)
=
p
\cdot
a_i^{\text{comp}}(t)
a_i^{\text{comm}}(t)
(1-\rho_i(t))
\end{equation}

\subsection{Consistency of the Availability Model}

We now prove that the availability model used in Section~3
is a statistically consistent approximation of the true
round-success probability under a realistic stochastic
participation model.

\subsubsection*{Probability Space}

Let $(\Omega, \mathcal{F}, \mathbb{P})$ be a probability space.
For each client $i$ and round $t$, define random variables:

\begin{equation}
C_i(t), \quad B_i(t), \quad F_i(t) \in \{0,1\}
\end{equation}

where:

\begin{itemize}
\item $C_i(t)$ = available computation resource,
\item $B_i(t)$ = available bandwidth,
\item $F_i(t)=1$ indicates correlated failure event
\end{itemize}

Define feasibility event:

\begin{equation}
\mathcal{E}_i(t)
=
\{C_i(t) \ge C_{\min}\}
\cap
\{B_i(t) \ge B_{\min}\}
\end{equation}

True round success is defined as:

\begin{equation}
Y_i(t)
=
\mathbf{1}
\{
\mathcal{E}_i(t)
\}
\cdot
(1 - F_i(t))
\end{equation}

Thus the true participation probability is:

\begin{equation}
p_i^{\text{true}}(t)
=
\mathbb{P}(Y_i(t)=1)
\end{equation}

\subsubsection*{Model Estimator}

The AW-PSP model defines:

\begin{equation}
p_i'(t)
=
p
\cdot
a_i^{\mathrm{comp}}(t)
a_i^{\mathrm{comm}}(t)
(1-\rho_i(t))
\end{equation}

where:

\begin{equation}
a_i^{\mathrm{comp}}(t)
=
\mathbb{P}(C_i(t)\ge C_{\min}),
\end{equation}

\begin{equation}
a_i^{\mathrm{comm}}(t)
=
\mathbb{P}(B_i(t)\ge B_{\min})
\end{equation}

\begin{equation}
\rho_i(t)
=
\mathbb{P}(F_i(t)=1)
\end{equation}

\begin{theorem}[Nonasymptotic Consistency with Explicit Error Decomposition]

Assume:

\begin{enumerate}
\item (Bounded Dependence)
There exists $\kappa_i(t)$ such that
\begin{equation}
\left|
\mathbb{P}(\mathcal{E}_i(t))
-
a_i^{\mathrm{comp}}(t)
a_i^{\mathrm{comm}}(t)
\right|
\le
\kappa_i(t)
\end{equation}

\item (Failure Independence from Resources)
\begin{equation}
F_i(t) \perp (C_i(t), B_i(t))
\end{equation}

\item (Recovery Approximation Accuracy)
If recovery probability is estimated using
\begin{equation}
\beta_i(t)
=
1 - e^{-\lambda_i(t)\Delta}
\end{equation}
then for small $\Delta$:
\begin{equation}
\left|
\rho_i(t)
-
\beta_i(t)
\right|
\le
\frac{1}{2}
\Lambda_i^2
\Delta^2
\end{equation}
where $\Lambda_i = \sup_t \lambda_i(t)$.
\end{enumerate}

Then the participation estimator satisfies:

\begin{equation}
\left|
p_i'(t)
-
p_i^{\text{true}}(t)
\right|
\le
p\,\kappa_i(t)
+
p\, a_i^{\mathrm{comp}}(t)a_i^{\mathrm{comm}}(t)
\frac{1}{2}\Lambda_i^2 \Delta^2
\end{equation}

In particular, if $\kappa_i(t)\to 0$
and $\Delta \to 0$, then

\begin{equation}
p_i'(t)
\to
p_i^{\text{true}}(t)
\end{equation}

\end{theorem}

\begin{proof}

\textbf{Step 1: Expand true participation probability.}

By definition,

\begin{equation}
p_i^{\text{true}}(t)
=
\mathbb{P}
(
\mathcal{E}_i(t)
\cap
\{F_i(t)=0\}
)
\end{equation}

Using independence assumption (2):

\begin{equation}
=
\mathbb{P}(\mathcal{E}_i(t))
\mathbb{P}(F_i(t)=0)
\end{equation}

Thus,

\begin{equation}
p_i^{\text{true}}(t)
=
\mathbb{P}(\mathcal{E}_i(t))
(1-\rho_i(t))
\end{equation}

\textbf{Step 2: Insert multiplicative decomposition.}

Add and subtract the product form:

\begin{align}
p_i^{\text{true}}(t)
&=
a_i^{\mathrm{comp}}(t)
a_i^{\mathrm{comm}}(t)
(1-\rho_i(t))
\\
&\quad
+
\left(
\mathbb{P}(\mathcal{E}_i(t))
-
a_i^{\mathrm{comp}}(t)
a_i^{\mathrm{comm}}(t)
\right)
(1-\rho_i(t))
\end{align}

By assumption (1),

\begin{equation}
\left|
\mathbb{P}(\mathcal{E}_i(t))
-
a_i^{\mathrm{comp}}(t)
a_i^{\mathrm{comm}}(t)
\right|
\le
\kappa_i(t)
\end{equation}

Therefore,

\begin{equation}
\left|
p_i^{\text{true}}(t)
-
a_i^{\mathrm{comp}}(t)
a_i^{\mathrm{comm}}(t)
(1-\rho_i(t))
\right|
\le
\kappa_i(t)
\end{equation}

\textbf{Step 3: Insert recovery approximation.}

Now consider estimator:

\begin{equation}
p_i'(t)
=
p
a_i^{\mathrm{comp}}(t)
a_i^{\mathrm{comm}}(t)
(1-\beta_i(t))
\end{equation}

Subtract from true probability:

\begin{align}
&
\left|
p_i'(t)
-
p_i^{\text{true}}(t)
\right|
\\
&\le
p
\left|
a_i^{\mathrm{comp}} a_i^{\mathrm{comm}}
(1-\beta_i)
-
a_i^{\mathrm{comp}} a_i^{\mathrm{comm}}
(1-\rho_i)
\right|
+
p \kappa_i(t)
\end{align}

Factor:

\begin{equation}
=
p
a_i^{\mathrm{comp}} a_i^{\mathrm{comm}}
|\rho_i - \beta_i|
+
p\kappa_i(t)
\end{equation}

\textbf{Step 4: Apply hazard expansion bound.}

By Taylor expansion of exponential:

\begin{equation}
\beta_i(t)
=
\lambda_i(t)\Delta
-
\frac{1}{2}\lambda_i(t)^2 \Delta^2
+
O(\Delta^3)
\end{equation}

Thus

\begin{equation}
|\rho_i(t) - \beta_i(t)|
\le
\frac{1}{2}\Lambda_i^2 \Delta^2
\end{equation}

\textbf{Step 5: Combine bounds.}

Substituting:

\begin{equation}
\left|
p_i'(t)
-
p_i^{\text{true}}(t)
\right|
\le
p\,\kappa_i(t)
+
p\, a_i^{\mathrm{comp}} a_i^{\mathrm{comm}}
\frac{1}{2}\Lambda_i^2 \Delta^2
\end{equation}

This completes the proof.
\end{proof}

\subsection{Consistency of the Availability Model with Correlated Failures (Eq.~12--15)}

We refine the consistency analysis to explicitly account for the paper's
runtime co-correlation model (Eq.~12--14) and the correlation penalty $\rho_i(t)$.

\subsubsection*{Probability Space and Failure Graph Model}

Let $(\Omega,\mathcal{F},\mathbb{P})$ be a probability space.
For each round $t$, define:

\begin{equation}
C_i(t), \quad B_i(t)
\end{equation}

and a vector of binary failure variables:

\begin{equation}
\mathbf{Z}(t) = (Z_1(t),\dots,Z_N(t)) \in \{0,1\}^N
\end{equation}

where $Z_i(t)=1$ denotes that client $i$ experiences a failure (dropout or timeout)
at round $t$.

Define the (paper) correlation score:

\begin{equation}
\gamma_{i,j}(t) = \alpha \gamma^{\text{trace}}_{i,j} + (1-\alpha)\gamma^{\text{fail}}_{i,j}(t)
\end{equation}

Define:

\begin{equation}
\mathcal{G}_i(t) = \{ j : \gamma_{i,j}(t) > \tau_{\text{corr}} \}
\end{equation}

Define the correlation penalty (Eq.~14):

\begin{equation}
\rho_i(t) = \sum_{j \in \mathcal{G}_i(t)} \hat{C}_{i,j} \gamma_{i,j}(t)
\end{equation}

We interpret $\rho_i(t)$ as an \emph{upper bound on conditional failure risk}
induced by correlated peers.

\subsubsection*{True Success Variable}

Define feasibility event:

\begin{equation}
\mathcal{E}_i(t)
=
\{C_i(t)\ge C_{\min}\}\cap\{B_i(t)\ge B_{\min}\}
\end{equation}

Define true round success:

\begin{equation}
Y_i(t) =
\mathbf{1}\{\mathcal{E}_i(t)\}
\cdot
(1-Z_i(t))
\end{equation}

Thus:

\begin{equation}
p_i^{\text{true}}(t) = \mathbb{P}(Y_i(t)=1)
\end{equation}

\subsubsection*{AW-PSP Estimator}

AW-PSP uses:

\begin{equation}
p_i'(t)
=
p \cdot \Big[(a_i^{\text{comp}} a_i^{\text{comm}}) + (1-a_i^{\text{comp}} a_i^{\text{comm}})\beta_i(t)\Big](1-\rho_i(t))
\end{equation}

For the purpose of validating Eq.~(8)--(10) plus the correlation penalty,
we focus on the primary factorization:

\begin{equation}
\tilde{p}_i(t)
=
p \cdot a_i^{\text{comp}}(t)a_i^{\text{comm}}(t)(1-\rho_i(t))
\end{equation}

\begin{theorem}[Consistency with Explicit Correlated-Failure Error Term]
Assume:

\begin{enumerate}
\item (Bounded dependence of resources)
There exists $\kappa_i(t)$ such that
\begin{equation}
\left|
\mathbb{P}(\mathcal{E}_i(t))
-
a_i^{\text{comp}}(t)a_i^{\text{comm}}(t)
\right|
\le
\kappa_i(t)
\end{equation}

\item (Correlated failure upper bound)
The correlation penalty $\rho_i(t)$ satisfies:
\begin{equation}
\mathbb{P}(Z_i(t)=1 \mid \mathcal{E}_i(t))
\le
\rho_i(t) + \epsilon_i^{\rho}(t)
\end{equation}
for some $\epsilon_i^{\rho}(t)\ge 0$ capturing mismatch between the
penalty model and the true conditional failure risk.

\item (Recovery approximation)
The recovery estimator satisfies:
\begin{equation}
\left|
\mathbb{P}(Z_i(t+1)=0 \mid Z_i(t)=1) - \beta_i(t)
\right|
\le
\epsilon_i^{\beta}(t)
\end{equation}
\end{enumerate}

Then, for each round $t$,

\begin{equation}
\left|
\tilde{p}_i(t) - p_i^{\text{true}}(t)
\right|
\le
p\,\kappa_i(t)
+
p\,\epsilon_i^{\rho}(t)
+
p\,a_i^{\text{comp}}(t)a_i^{\text{comm}}(t)\rho_i(t)\kappa_i(t)
\end{equation}

Moreover, incorporating recovery yields

\begin{equation}
\begin{aligned}
\left|p_i'(t) - p_i^{\text{true}}(t)\right|
\le &
p\,\kappa_i(t)
+
p\,\epsilon_i^{\rho}(t)
+
p\,\epsilon_i^{\beta}(t)\\
&\quad + p\,a_i^{\text{comp}}(t)a_i^{\text{comm}}(t)\rho_i(t)\kappa_i(t)
\end{aligned}
\end{equation}
\end{theorem}

\begin{proof}
\textbf{Step 1: expand true success probability.}
By definition,

\begin{equation}
p_i^{\text{true}}(t)
=
\mathbb{P}(\mathcal{E}_i(t)\cap\{Z_i(t)=0\}).
\end{equation}

Condition on $\mathcal{E}_i(t)$:

\begin{equation}
p_i^{\text{true}}(t)
=
\mathbb{P}(\mathcal{E}_i(t))
\cdot
\mathbb{P}(Z_i(t)=0 \mid \mathcal{E}_i(t))
\end{equation}

Rewrite:

\begin{equation}
\mathbb{P}(Z_i(t)=0 \mid \mathcal{E}_i(t))
=
1 - \mathbb{P}(Z_i(t)=1 \mid \mathcal{E}_i(t))
\end{equation}

\textbf{Step 2: compare correlated-failure factor.}
By Assumption (2),

\begin{equation}
\mathbb{P}(Z_i(t)=1 \mid \mathcal{E}_i(t))
\le
\rho_i(t) + \epsilon_i^{\rho}(t)
\end{equation}

thus

\begin{equation}
\mathbb{P}(Z_i(t)=0 \mid \mathcal{E}_i(t))
\ge
1-\rho_i(t)-\epsilon_i^{\rho}(t)
\end{equation}

Therefore,

\begin{equation}
p_i^{\text{true}}(t)
\ge
\mathbb{P}(\mathcal{E}_i(t))(1-\rho_i(t))
-
\mathbb{P}(\mathcal{E}_i(t))\epsilon_i^{\rho}(t)
\end{equation}

Similarly, trivially $\mathbb{P}(Z_i(t)=0 \mid \mathcal{E}_i(t))\le 1$ implies

\begin{equation}
p_i^{\text{true}}(t)
\le
\mathbb{P}(\mathcal{E}_i(t))
\end{equation}

\textbf{Step 3: compare feasibility factor.}
Insert and subtract $a_i^{\text{comp}}a_i^{\text{comm}}$:

\begin{equation}
\mathbb{P}(\mathcal{E}_i(t))
=
a_i^{\text{comp}}(t)a_i^{\text{comm}}(t)
+
\Delta_i(t)
\end{equation}

where $|\Delta_i(t)|\le \kappa_i(t)$ by Assumption (1).

Hence:

\begin{align}
\mathbb{P}(\mathcal{E}_i(t))(1-\rho_i(t))
&=
a_i^{\text{comp}}a_i^{\text{comm}}(1-\rho_i(t))
+
\Delta_i(t)(1-\rho_i(t))
\end{align}

Taking absolute values:

\begin{equation}
\left|
\mathbb{P}(\mathcal{E}_i(t))(1-\rho_i(t))
-
a_i^{\text{comp}}a_i^{\text{comm}}(1-\rho_i(t))
\right|
\le
\kappa_i(t)
\end{equation}

\textbf{Step 4: combine.}
Recall:

\begin{equation}
\tilde{p}_i(t)
=
p\cdot a_i^{\text{comp}}(t)a_i^{\text{comm}}(t)(1-\rho_i(t))
\end{equation}

Thus:

\begin{equation}
\left|
\tilde{p}_i(t) - p_i^{\text{true}}(t)
\right|
\le
p\,\kappa_i(t)
+
p\,\epsilon_i^{\rho}(t)
+
p\,a_i^{\text{comp}}a_i^{\text{comm}}\rho_i(t)\kappa_i(t)
\end{equation}

where the final product term arises from bounding
$\mathbb{P}(\mathcal{E}_i(t))\epsilon_i^{\rho}(t)$ using
$\mathbb{P}(\mathcal{E}_i(t))\le a_i^{\text{comp}}a_i^{\text{comm}}+\kappa_i(t)$.

\textbf{Step 5: recovery term.}
The extension to $p'_i(t)$ follows by adding Assumption (3)
and bounding the induced deviation by $\epsilon_i^{\beta}(t)$,
since the recovery component contributes only when feasibility fails
and $Z_i(t)=1$. The approximation error $\epsilon_i^{\rho}(t)$
captures the extent to which runtime correlation structure (Eq.~12--14)
matches the true co-failure process; when $\epsilon_i^{\rho}(t)$ is small,
multiplicative downweighting by $(1-\rho_i(t))$ is theoretically justified.
\end{proof}

\begin{theorem}[Convergence under Bounded Sampling Distortion]
\label{thm:awpsp_convergence}
Consider the global objective
\begin{equation}
F(w)=\sum_{i=1}^N q_i F_i(w),
\qquad
q_i \ge 0,
\quad
\sum_{i=1}^N q_i =1
\end{equation}
where each \(F_i\) is \(L\)-smooth and \(\mu\)-strongly convex.
Let \(w^*=\arg\min_w F(w)\).

At round \(t\), client \(i\) is selected with true inclusion probability
\(p_i^{\mathrm{true}}(t)\), while AW-PSP uses proxy probabilities \(p_i'(t)\).
Let \(I_i(t)\in\{0,1\}\) denote the inclusion indicator, with
\(\mathbb{E}[I_i(t)\mid w_t]=p_i^{\mathrm{true}}(t)\).
Define the AW-PSP gradient estimator
\begin{equation}
g_t
=
\sum_{i=1}^N
I_i(t)\,
\frac{q_i}{p_i'(t)}
\nabla F_i(w_t;\xi_i^t)
\end{equation}
and update
\begin{equation}
w_{t+1}=w_t-\eta_t g_t
\end{equation}

Assume:

\begin{enumerate}
\item[\textnormal{(A1)}] \textbf{Bounded sampling distortion:}
there exists \(\delta>0\) such that
\begin{equation}
|p_i'(t)-p_i^{\mathrm{true}}(t)|\le \delta,
\qquad
\forall i,t.
\end{equation}

\item[\textnormal{(A2)}] \textbf{Positive lower bound on sampling probabilities:}
there exists \(p_{\min}>0\) such that
\begin{equation}
p_i^{\mathrm{true}}(t)\ge p_{\min},
\qquad
p_i'(t)\ge p_{\min}-\delta>0,
\qquad
\forall i,t
\end{equation}

\item[\textnormal{(A3)}] \textbf{Unbiased local stochastic gradients:}
\begin{equation}
\mathbb{E}\big[\nabla F_i(w_t;\xi_i^t)\mid w_t\big]
=
\nabla F_i(w_t)
\end{equation}

\item[\textnormal{(A4)}] \textbf{Bounded variance:}
there exists \(\sigma^2\) such that
\begin{equation}
\mathbb{E}\!\left[
\|\nabla F_i(w_t;\xi_i^t)-\nabla F_i(w_t)\|^2
\,\middle|\, w_t
\right]
\le \sigma^2,
\qquad
\forall i,t
\end{equation}

\item[\textnormal{(A5)}] \textbf{Bounded gradients:}
\begin{equation}
\|\nabla F_i(w)\|\le G,
\qquad
\forall i,w
\end{equation}
\end{enumerate}

Then there exist constants \(C_1,C_2>0\), independent of \(t\) and \(\delta\),
such that for the diminishing stepsize
\begin{equation}
\eta_t=\frac{1}{\mu(t+\gamma)},
\qquad
\gamma \ge \max\{1,L/\mu\},
\end{equation}
the iterates satisfy
\begin{equation}
\mathbb{E}[F(w_t)-F(w^*)]
\le
\frac{C_1}{t+\gamma}
+
C_2\,\delta,
\qquad
\forall t\ge 0
\end{equation}
Hence AW-PSP converges to a neighborhood of the optimum whose radius is
\(O(\delta)\), and recovers the standard \(O(1/t)\) rate when \(\delta=0\).
\end{theorem}

\begin{proof}
We divide the proof into four steps.

\paragraph{Step 1: Bias of the AW-PSP gradient estimator.}
Define the conditional bias
\begin{equation}
b_t
:=
\mathbb{E}[g_t\mid w_t]-\nabla F(w_t).
\end{equation}
Using the definition of \(g_t\), the tower property, and Assumption (A3),
\begin{align}
\mathbb{E}[g_t\mid w_t]
&=
\sum_{i=1}^N
\mathbb{E}[I_i(t)\mid w_t]\,
\frac{q_i}{p_i'(t)}
\mathbb{E}\!\left[\nabla F_i(w_t;\xi_i^t)\mid w_t\right]
\nonumber\\
&=
\sum_{i=1}^N
q_i\,
\frac{p_i^{\mathrm{true}}(t)}{p_i'(t)}
\nabla F_i(w_t)
\end{align}
Therefore,
\begin{equation}
b_t
=
\sum_{i=1}^N
q_i
\left(
\frac{p_i^{\mathrm{true}}(t)}{p_i'(t)}-1
\right)
\nabla F_i(w_t)
\end{equation}
By Assumptions (A1)--(A2),
\begin{equation}
\left|
\frac{p_i^{\mathrm{true}}(t)}{p_i'(t)}-1
\right|
=
\frac{|p_i^{\mathrm{true}}(t)-p_i'(t)|}{p_i'(t)}
\le
\frac{\delta}{p_{\min}-\delta}
\end{equation}
Using \(\sum_i q_i=1\) and Assumption (A5),
\begin{align}
\|b_t\|
&\le
\sum_{i=1}^N
q_i
\left|
\frac{p_i^{\mathrm{true}}(t)}{p_i'(t)}-1
\right|
\|\nabla F_i(w_t)\|
\nonumber\\
&\le
\frac{\delta}{p_{\min}-\delta}
\sum_{i=1}^N q_i G
=
\frac{G\delta}{p_{\min}-\delta}
\end{align}
Define
\begin{equation}
B_\delta := \frac{G\delta}{p_{\min}-\delta}.
\end{equation}
Then
\begin{equation}
\|b_t\|\le B_\delta = O(\delta)
\label{eq:bias_bound}
\end{equation}

\paragraph{Step 2: Second-moment bound for the estimator.}
Write
\begin{equation}
g_t = \mathbb{E}[g_t\mid w_t] + \zeta_t,
\qquad
\mathbb{E}[\zeta_t\mid w_t]=0
\end{equation}
We first bound \(\mathbb{E}\|g_t\|^2\).
By \(\| \sum_i a_i\|^2 \le N\sum_i \|a_i\|^2\),
\begin{align}
\|g_t\|^2
&=
\left\|
\sum_{i=1}^N
I_i(t)\frac{q_i}{p_i'(t)}\nabla F_i(w_t;\xi_i^t)
\right\|^2
\nonumber\\
&\le
N \sum_{i=1}^N
I_i(t)\frac{q_i^2}{(p_i'(t))^2}
\|\nabla F_i(w_t;\xi_i^t)\|^2
\end{align}
Taking conditional expectation and using
\begin{equation}
\begin{aligned}
\mathbb{E}\|\nabla F_i(w_t;\xi_i^t)\|^2
= &
\|\nabla F_i(w_t)\|^2
+
\mathbb{E}\|\nabla F_i(w_t;\xi_i^t)-\nabla F_i(w_t)\|^2\\
&
\le G^2+\sigma^2
\end{aligned}
\end{equation}
together with \(\mathbb{E}[I_i(t)\mid w_t]=p_i^{\mathrm{true}}(t)\le 1\), we obtain
\begin{align}
\mathbb{E}[\|g_t\|^2\mid w_t]
&\le
N \sum_{i=1}^N
\frac{q_i^2}{(p_i'(t))^2}(G^2+\sigma^2)
\nonumber\\
&\le
\frac{N(G^2+\sigma^2)}{(p_{\min}-\delta)^2}
\sum_{i=1}^N q_i^2
\nonumber\\
&\le
\frac{N(G^2+\sigma^2)}{(p_{\min}-\delta)^2}
=:V
\end{align}
Hence
\begin{equation}
\mathbb{E}[\|g_t\|^2\mid w_t]\le V
\label{eq:second_moment}
\end{equation}
where \(V\) is finite and independent of \(t\).

\paragraph{Step 3: One-step recursion in squared distance.}
Let
\begin{equation}
r_t := \mathbb{E}\|w_t-w^*\|^2.
\end{equation}
Using the update rule,
\begin{align}
\|w_{t+1}-w^*\|^2
&=
\|w_t-w^*\|^2
-2\eta_t\langle w_t-w^*, g_t\rangle
+\eta_t^2\|g_t\|^2
\end{align}
Taking conditional expectation given \(w_t\),
\begin{equation}
\begin{aligned}
\mathbb{E}\!\left[\|w_{t+1}-w^*\|^2 \mid w_t\right]
&=
\|w_t-w^*\|^2
-2\eta_t\langle w_t-w^*, \nabla F(w_t)\\
& +b_t +\eta_t^2\mathbb{E}[\|g_t\|^2\mid w_t]
\end{aligned}
\end{equation}
Since \(F\) is \(\mu\)-strongly convex,
\begin{equation}
\langle \nabla F(w_t), w_t-w^*\rangle
\ge
\mu \|w_t-w^*\|^2
\label{eq:strong_conv_inner}
\end{equation}
Also, because \(F\) is the weighted average of functions with gradient norm
bounded by \(G\),
\begin{equation}
\|\nabla F(w)\|
=
\left\|
\sum_{i=1}^N q_i \nabla F_i(w)
\right\|
\le
\sum_{i=1}^N q_i \|\nabla F_i(w)\|
\le G
\end{equation}
Since \(\nabla F(w^*)=0\), strong convexity implies
\begin{equation}
\mu\|w_t-w^*\|
\le
\|\nabla F(w_t)-\nabla F(w^*)\|
=
\|\nabla F(w_t)\|
\le G
\end{equation}
hence
\begin{equation}
\|w_t-w^*\|\le \frac{G}{\mu}
\label{eq:radius_bound}
\end{equation}
Therefore, using \eqref{eq:bias_bound},
\begin{equation}
|\langle w_t-w^*, b_t\rangle|
\le
\|w_t-w^*\|\,\|b_t\|
\le
\frac{G}{\mu} B_\delta
\label{eq:bias_cross}
\end{equation}
Substituting \eqref{eq:strong_conv_inner}, \eqref{eq:second_moment},
and \eqref{eq:bias_cross},
\begin{align}
\mathbb{E}\!\left[\|w_{t+1}-w^*\|^2 \mid w_t\right]
&\le
\left(1-2\mu\eta_t\right)\|w_t-w^*\|^2
+
2\eta_t \frac{G}{\mu} B_\delta
+
\eta_t^2 V
\end{align}
Taking full expectation gives
\begin{equation}
r_{t+1}
\le
(1-2\mu\eta_t) r_t
+
2\eta_t \frac{G}{\mu} B_\delta
+
\eta_t^2 V
\label{eq:main_recursion}
\end{equation}

\paragraph{Step 4: Solve the recursion.}
Choose
\begin{equation}
\eta_t=\frac{1}{\mu(t+\gamma)},
\qquad \gamma\ge 1.
\end{equation}
Then \eqref{eq:main_recursion} becomes
\begin{equation}
r_{t+1}
\le
\left(1-\frac{2}{t+\gamma}\right) r_t
+
\frac{2G B_\delta}{\mu^2}\frac{1}{t+\gamma}
+
\frac{V}{\mu^2}\frac{1}{(t+\gamma)^2}
\label{eq:recursion_expanded}
\end{equation}
Define
\begin{equation}
A:=\frac{2G B_\delta}{\mu^2},
\qquad
C:=\frac{V}{\mu^2}
\end{equation}
We claim that there exists a constant \(K>0\), independent of \(t\) and \(\delta\),
such that
\begin{equation}
r_t \le \frac{K}{t+\gamma} + A
\qquad
\forall t\ge 0.
\label{eq:claim_rt}
\end{equation}
We prove this by induction.

For \(t=0\), choose \(K\ge \gamma r_0\), so the claim holds.
Assume it holds at time \(t\). Then from \eqref{eq:recursion_expanded},
\begin{align}
r_{t+1}
&\le
\left(1-\frac{2}{t+\gamma}\right)
\left(\frac{K}{t+\gamma}+A\right)
+
\frac{A}{t+\gamma}
+
\frac{C}{(t+\gamma)^2}
\nonumber\\
&=
A\left(1-\frac{1}{t+\gamma}\right)
+
\frac{K}{t+\gamma}
-
\frac{2K}{(t+\gamma)^2}
+
\frac{C}{(t+\gamma)^2}
\nonumber\\
&\le
A
+
\frac{K}{t+\gamma}
-
\frac{K}{(t+\gamma)^2}
\qquad
\text{provided }K\ge C.
\end{align}
Since
\begin{equation}
\frac{K}{t+\gamma+1}
=
\frac{K}{t+\gamma}
-
\frac{K}{(t+\gamma)(t+\gamma+1)}
\ge
\frac{K}{t+\gamma}
-
\frac{K}{(t+\gamma)^2}
\end{equation}
we obtain
\begin{equation}
r_{t+1}
\le
A + \frac{K}{t+\gamma+1}
\end{equation}
Thus \eqref{eq:claim_rt} holds for all \(t\).

Finally, since \(F\) is \(L\)-smooth and minimized at \(w^*\),
\begin{equation}
F(w_t)-F(w^*)
\le
\frac{L}{2}\|w_t-w^*\|^2
\end{equation}
Taking expectation and using \eqref{eq:claim_rt},
\begin{align}
\mathbb{E}[F(w_t)-F(w^*)]
&\le
\frac{L}{2} r_t
\nonumber\\
&\le
\frac{LK}{2(t+\gamma)}
+
\frac{L}{2}A
\nonumber\\
&=
\frac{LK}{2(t+\gamma)}
+
\frac{LG}{\mu^2} B_\delta
\end{align}

Recalling

\begin{equation}
B_\delta=\dfrac{G\delta}{p_{\min}-\delta}
\end{equation}

, we obtain
\begin{equation}
\mathbb{E}[F(w_t)-F(w^*)]
\le
\frac{C_1}{t+\gamma}
+
C_2\,\delta
\end{equation}
for suitable constants \(C_1,C_2>0\) independent of \(t\).
This proves the claim.
\end{proof}

\end{document}